\documentclass{article}

\usepackage{jmlr2e}

\usepackage{hyperref}       % hyperlinks
\usepackage{url}            % simple URL typesetting
\usepackage{amsfonts}       % blackboard math symbols
\usepackage{nicefrac}       % compact symbols for 1/2, etc.
\usepackage{microtype}      % microtypography
\usepackage{lipsum}
\usepackage{graphicx}       % graphics

\usepackage{mathrsfs}
\usepackage{amsmath}
\usepackage{amssymb}
\usepackage{caption}
\usepackage{subcaption}
\usepackage{float}
\usepackage{tikz}
\usepackage{algorithm}
\usepackage{algpseudocode}
\usetikzlibrary{calc, matrix, arrows, shapes, positioning, fit, shapes.misc, shapes.geometric, calc, decorations.pathreplacing}
\newcommand{\norm}[1]{\Vert{#1}\Vert}
\usepackage{multirow}
\usepackage{dblfloatfix}
\usepackage{tikz}
\usepackage{kbordermatrix}
\usetikzlibrary{shadows,shadings,shapes.symbols, calc}

\newtheorem{assump}{Assumption}

\newcommand{\tbl}[1]{\textcolor{uvablue!70!white}{#1}}
\newcommand{\tor}[1]{\textcolor{uvaorange}{#1}}

\newcommand{\RR}{\mathbb{R}}
\newcommand{\PP}{\mathbb{P}}
\newcommand{\EE}{\mathbb{E}}
\newcommand{\NN}{\mathbb{N}}

\newcommand{\hol}[1]{\mathcal{H}(#1)}

\definecolor{uvablue}{RGB}{12, 39, 92}
\definecolor{uvaorange}{RGB}{243, 115, 33}
  
\title{Perfect clustering in nonuniform hypergraphs
}

\usepackage{lastpage}
\jmlrheading{xx}{2025}{1-\pageref{LastPage}}{4/25}{x/xx}{xx-0000}{Ga-Ming Angus Chan and Zachary Lubberts}

\ShortHeadings{Perfect Clustering in Nonuniform Hypergraphs}{Chan and Lubberts}
\firstpageno{1}

\begin{document}

\title{Perfect Clustering in Nonuniform Hypergraphs}

\author{\name Ga-Ming (Angus) Chan \email gc8ev@virginia.edu \\
       \addr Department of Statistics\\
       University of Virginia\\
       Charlottesville, VA 22904, USA
       \AND
       \name Zachary Lubberts \email zlubberts@virginia.edu \\
       \addr Department of Statistics\\
       University of Virginia\\
       Charlottesville, VA 22904, USA}

\editor{TBD}

\maketitle

\begin{abstract}
  While there has been tremendous activity in the area of statistical network inference on graphs, hypergraphs have not enjoyed the same attention, on account of their relative complexity and the lack of tractable statistical models. We introduce a hyper-edge-centric model for analyzing hypergraphs, called the interaction hypergraph, which models natural sampling methods for hypergraphs in neuroscience and communication networks, and accommodates interactions involving different numbers of entities. We define latent embeddings for the interactions in such a network, and analyze their estimators. In particular, we show that a spectral estimate of the interaction latent positions can achieve perfect clustering once enough interactions are observed.
\end{abstract}


\section{Introduction}
\label{sec:intro}

Networks have received widespread recognition as a useful tool for modeling complex systems throughout the past decade. However, much of the focus has been on studying pairwise connections. Extending these methods to the case where connections can involve more than two nodes, in \emph{hypernetworks} or \emph{hypergraphs}, is an important step for the area of statistical network inference, since it allows for higher-order relationships to be captured. However, these models have not received the same attention due to their novelty and mathematical complexity. This is to the detriment of the practitioner, since in many real systems, the relationships are not pairwise. In neuroscience, such as the celebrated whole-brain connectomes \cite{winding2023connectome}, synapses correspond to discrete spatial locations in the brain where several neurons are in close proximity. Email communication networks and coauthorship networks also involve interactions between groups of users. While reducing these relationships into collections of pairwise edges is commonplace, this loses much of the complexity and information of the original data. 

From our perspective, there are four issues with the traditional pairwise approaches to such systems:
\begin{itemize}
    \item The meanings of the relationships change: A single email involving 4 people is quite different from 6 sets of pairwise email exchanges. A synaptic firing that affects two other neurons simultaneously is mechanically quite different from a pair of synapses that can operate independently.
    \item The sampling methods are not representative of reality: When we choose a window of time to construct an email network, we select a subset of all \emph{emails}, not all users. When we select a spatial brain region, we select a subset of all \emph{synapses}, not all neurons, since the latter are spatially disparate. \citep{crane2018edge}
    \item Dependence relationships within interactions are ignored: Replacing a hyper-edge with $k$ vertices with $\binom{k}{2}$ pairwise edges results in a clique, and models which suppose independent edge formations will not capture this.
    \item Edge multiplicity is the rule, rather than the exception: We expect to see email threads involving several interactions between the same 5 people. In the Drosophila connectome, synapses among the same neurons occur multiple times.
\end{itemize}

In what follows, we introduce the idea of \emph{interaction hypergraphs} to model hypergraphs. Contrary to the node-centric view of sampling, interaction hypergraphs consider an interaction-centric view where the sampling units are interactions. For the email network example, we might consider a sample of emails: If we treat each email as an interaction, the sampling units are interactions and the users involved in an email are the vertices comprising that interaction. Since we think of interactions as the sampling units, we want hyper-edge exchangeability, and this necessitates the possibility of repeated edges (just as we can, and often do, have several emails among the same set of users during a given time period). This motivates our usage of the term \emph{interactions}, making a clear distinction between the present setting and that of traditional network models, which assume each edge can appear only once. Weighted network approaches may address this point in part, but they lose the fundamentally discrete nature of repeated interactions, and still fail to address the previous points. The interaction hypergraph approach preserves the information contained in the interactions; reflects realistic sampling methods; naturally encodes the dependence that appears within interactions; and allows for edge multiplicity, and as such, it addresses all of the issues mentioned in the previous paragraph. 

Previous work on hypergraphs has focused on connectivity \citep{hu2012algebraic,hu2015laplacian,xie2016spectral,banerjee2021spectrum}, Hamiltonian cycles \citep{gu2022hamiltonian}, dynamical systems \citep{boccaletti2023structure}, colorings and clique numbers \citep{cooper2012spectra,rota2009new,xie2015clique}, and other structural properties. Moreover, many of these papers assume that the hypergraphs are \emph{uniform}, so that all interactions have the same size. The few, very recent exceptions which have considered general hypergraphs \citep{sun2019spectral,cardoso2021principal,cardoso2021general,lin2020alpha} are understandably narrow in scope, focusing on either the spectral radius or principal eigenvector, and in a deterministic setting. As such, while these results may provide some ideas for statistical inference on general hypergraphs, there is still a good deal of work to be done. We also note the use of tensor-based methods for studying hypergraphs. These necessarily suppose uniformity, and often assume asymmetry: see \citet{luo2022tensor,han2022exact,agterberg2022estimating}. 

An important exception to these comments can be found in \citet{chodrow2023nonbacktracking}, which considers the pointed line graph model, and does not require tensor or degree uniformity. While the models for hypergraph generation bear a surface resemblance, our model allows for interaction multiplicity, since it imposes interaction exchangeability. Moreover, we do not suppose a particular probability distribution for the sizes of different interactions, and our estimation procedure is quite different.

The structure of this paper is as follows. In Section~\ref{S:model}, we will go over our notations and model. Then in Section~\ref{sec:theory}, we discuss the properties of our proposed spectral method. Sections~\ref{sec:simulation} and \ref{sec:realdata} verify our results through generated and real data respectively. All proofs, as well as additional simulation results, may be found in the Appendix.

\section{Model and notation}
\label{S:model}

In this section, we detail the notations we will use throughout the paper and set down the exact form of our model for interaction hypergraphs. We propose a version of the model with some generality, but we prove theoretical results for a more restricted class.

\subsection{Notation}
We use $\PP$ to denote probabilities, $\EE$ for expectations. For a positive integer $n$, $[n]$ denotes the set $\{1,2,\dots,n\}$. For a vector $v$, $\|v\|$ denotes its Euclidean length. For an $n\times m$ real-valued matrix $M$, $M^\top$ denotes its transpose, $\norm{M}_F$ denotes its Frobenius norm, $\norm{M}$ or $\norm{M}_2$ denotes its spectral norm, and $\lambda_1(M), \lambda_2(M), \dots, \lambda_d(M)$ and $\sigma_1(M), \sigma_2(M), \dots, \sigma_d(M)$ denote the eigenvalues and singular values of $M$ respectively, in descending order (i.e. $\lambda_1(M)\geq\lambda_2(M)\geq\dots\geq \lambda_d(M)$, and $\sigma_1(M)\geq\sigma_2(M)\geq\dots\geq \sigma_d(M)$). We use $I_k$ to denote the identity matrix of order $k$, and we use $J_{k_1,k_2}, 0_{k_1,k_2}$ to denote the $k_1\times k_2$ matrices of all ones and all zeroes, respectively. We will write $J_{k_1,k_1}$ as $J_{k_1}$.

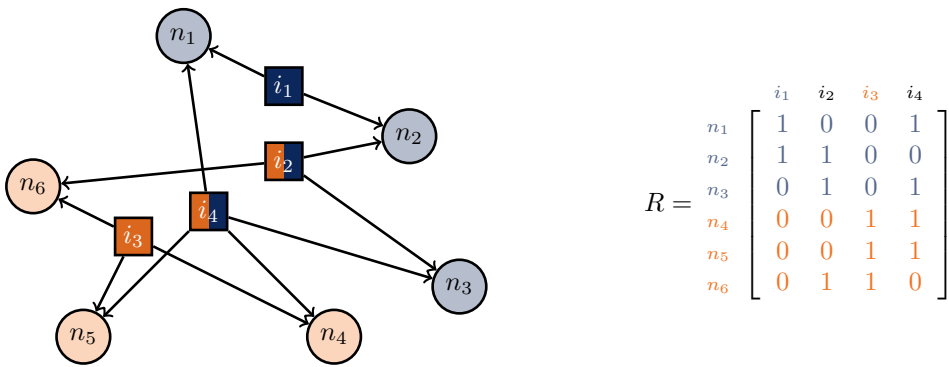
\begin{figure}
\centering
    \begin{minipage}[c]{0.475\textwidth}
      \centering
      \begin{tikzpicture}[
      vertex/.style = {circle, draw, inner sep=3pt, fill=white, scale=1},
      vertex1/.style = {vertex, fill=uvablue!30!white},
      vertex2/.style = {vertex, fill=uvaorange!30!white},
      interaction/.style = {draw, inner sep=3pt, fill=white, scale=1,fill=white},
      fill fraction/.style n args={2}{path picture={
        \fill[#1] (path picture bounding box.south west) rectangle
        ($(path picture bounding box.north west)!#2!(path picture bounding box.north east)$);}},
      ]
      \node[vertex1,line width = 1pt] (n1) at (0,4) {$n_1$};
      \node[vertex1,line width = 1pt] (n2) at (3,8/3){$n_2$};
      \node[vertex1,line width = 1pt] (n3) at (11/3,2/3) {$n_3$};
      \node[vertex2,line width = 1pt] (n4) at (2,0){$n_4$};
      \node[vertex2,line width = 1pt] (n5) at (-4/3,0){$n_5$};
      \node[vertex2,line width = 1pt] (n6) at (-2,2){$n_6$};

      \node[interaction,line width = 1pt, fill = uvablue,text=white] (i1) at (4/3,10/3) {$i_1$};%n_1 and n_2
      \node[interaction,line width = 1pt, text=white,fill=uvablue, fill fraction={uvaorange!90!black}{0.5}] (i2) at (4/3,7/3) {$i_2$};%n_2,3 and 6
      \node[interaction,line width = 1pt, fill = uvaorange!90!black,text=white] (i3) at (-2/3,4/3) {$i_3$};%4,5,6
      \node[interaction,line width = 1pt, fill=uvablue, fill fraction={uvaorange!90!black}{0.5},text=white] (i4) at (1/3,5/3) {$i_4$};%1,3,4,5

      \draw[->,line width = 1pt] (i1) to (n2);
      \draw[->,line width = 1pt] (i1) to (n1);
      \draw[->,line width = 1pt] (i2) to (n2);
      \draw[->,line width = 1pt] (i2) to (n3);
      \draw[->,line width = 1pt] (i2) to (n6);
      \draw[->,line width = 1pt] (i3) to (n4);
      \draw[->,line width = 1pt] (i3) to (n5);
      \draw[->,line width = 1pt] (i3) to (n6);
      \draw[->,line width = 1pt] (i4) to (n1);
      \draw[->,line width = 1pt] (i4) to (n3);
      \draw[->,line width = 1pt] (i4) to (n4);
      \draw[->,line width = 1pt] (i4) to (n5);        
    \end{tikzpicture}
    \end{minipage}
    \begin{minipage}[c]{0.475\textwidth}
    \begin{equation*}
     R = \kbordermatrix{
      &\tbl{i_1}&i_2&\tor{i_3}&i_4\\
     \tbl{n_1}&\tbl{1}&\tbl{0}&\tbl{0}&\tbl{1}\\
     \tbl{n_2}&\tbl{1}&\tbl{1}&\tbl{0}&\tbl{0}\\
     \tbl{n_3}&\tbl{0}&\tbl{1}&\tbl{0}&\tbl{1}\\
     \tor{n_4}&\tor{0}&\tor{0}&\tor{1}&\tor{1}\\
     \tor{n_5}&\tor{0}&\tor{0}&\tor{1}&\tor{1}\\
     \tor{n_6}&\tor{0}&\tor{1}&\tor{1}&\tor{0}
     }
    \end{equation*}
    \end{minipage}
    % \begin{pmatrix}
     %     \blue 1&\blue 0&\blue 0&\blue 1\\
     %     \blue 1&\blue 1&\blue 0&\blue 0\\
     %     \blue 0&\blue 1&\blue 0&\blue 1\\
     %     \orange 0&\orange0&\orange1&\orange1\\
     %     \orange0&\orange0&\orange1&\orange1\\
     %     \orange0&\orange1&\orange1&\orange0\\
     % \end{pmatrix}
      \caption{Demonstration of interaction hypergraphs: visualization (left) and incidence matrix $R$ (right). Nodes are represented by the circles $v_1,\dots,v_6$, and the colors blue and orange denote their node classes. Sampled interactions are represented by squares $i_1,\dots,i_4$, where the arrows point to the nodes comprising each interaction. Pure and mixed interactions are denoted by color, corresponding to the nodes making up each interaction.}
      \label{fig:IH-Demo}
\end{figure}

\subsection{Interaction hypergraphs}

An \emph{interaction hypergraph} $H=(N,E)$ consists of a pair of a set of \emph{nodes} or \emph{vertices}, $N$, and a multiset of interactions, $E$. We fix the node set as $N=[n]=\{1,\ldots,n\}$, where $n$ is a positive integer, and $E=\{e_1,\ldots,e_m\}$, where $e_p\subseteq N$, and need not be distinct for different values of $1\leq p\leq m$. The \emph{node degree} of a node $v\in N$ is the number of interactions it belongs to, i.e., $\mathrm{degree}(v)=\#\{1\leq p\leq m: v\in e_p\},$ while the \emph{interaction degree} of an interaction $e\in E$ is the number of interactions it intersects with, $\mathrm{degree}(e)=\#\{1\leq p\leq m: e\cap e_p\neq \varnothing, e_p\neq e\}.$ Since we are in a hypergraph setting, we also define the \emph{interaction size}, which is the number of nodes comprising that interaction, $|e|=\#\{v\in N: v\in e\}.$ We show an example of this in the left panel of Figure~\ref{fig:IH-Demo}, where nodes are represented by circles, and interactions by squares. The colors of the nodes and interactions show the clusters to which the nodes belong, and the types of the interactions. We say that an interaction is \emph{pure} when it only contains nodes from one cluster, and otherwise it is \emph{mixed}.

Given a hypergraph $H=(N,E)$, we may associate to it an \emph{incidence matrix} $R\in\{0,1\}^{n\times m}$, where the $(i,p)$ entry of $R$, denoted by $R_{ip}$, equals $1$ if the $p$-th interaction involves node $i$ (i.e. $i\in e_p$), and $0$ otherwise, i.e.
\begin{equation*}
    R_{ip}=
    \begin{cases}
        1&\text{if }i\in e_p,\\
        0&\text{o.w.}
    \end{cases}
\end{equation*}
As interactions are assumed to be independent, the columns of $R$ are independent, but entries in the same column of $R$ are not independent. The right panel of Figure~\ref{fig:IH-Demo} shows the incidence matrix for this example.

Now we introduce the interaction hypergraph model:

\begin{definition}[Latent position interaction hypergraph]
Let $\omega: \mathcal{X}\times \mathcal{Y}\rightarrow [0,+\infty)$ be a nonnegative \emph{weight function}. For each $i \in N$, let $x_i\in\mathcal{X}$ be its associated \emph{node latent position}. For each $p\in [m]$, let $y_p\in\mathcal{Y}$ be the associated \emph{interaction latent position}, and let $k_p=|e_p|$ be the interaction size. Then we obtain a \emph{Latent Position Interaction Hypergraph (LPH)} by independently drawing interactions $e_p$ having size $k_p$, sampled non-uniformly without replacement for initial weights $\omega(x_i, y_p)$.
\end{definition}

\begin{remark}
To explain what is meant by ``sampled non-uniformly without replacement,'' let us consider a small example. Suppose interaction 1 has a latent position satisfying $\omega(x_1,y_1)=1, \omega(x_2,y_1)=2, \omega(x_3, y_1)=3,$ and size 2. We draw the first vertex according to a Multinomial trial with probability vector $(1,2,3)/6$. Once the first vertex has been selected, we condition on that event and draw the second vertex according to a Multinomial trial with conditional probability vector again proportional to the weights, so if vertex 1 were chosen first, we choose between vertices 2 and 3 with probabilities $2/5$ and $3/5$, and if vertex 3 were chosen first, we choose between vertices 1 and 2 with probabilities $1/3$ and $2/3$. Note that unlike the case with uniform sampling, the marginal probabilities for draws after the first do not agree with the probability distribution implied by the initial vector of weights: For example, vertex 3 has probability $3/6=1/2$ of being the first vertex drawn, but has probability $(1/6)(3/5)+(2/6)(3/4)=7/20$ of being the second vertex drawn. This effect becomes less significant as the number of possible objects becomes larger, so long as the weights are of comparable size.
\end{remark}

\begin{remark}
Latent position graphs (LPGs) \cite{hoff2002latent} enjoy a limited sense of node exchangeability, in the sense that the distribution of the generated network does not change if nodes with the same latent positions are permuted. However, edge exchangeability in this setting is impossible, since any particular edge appears at most once. In the LPH model, we have this same limited sense of node exchangeability, as well as an analogous notion of hyper-edge exchangeability, in that permutations of the interactions which only exchange interactions having the same latent positions and sizes will also preserve the distribution. This sense of hyper-edge exchangeability is important so long as we consider the \emph{incidence matrix} of the hypergraph, $R\in \{0,1\}^{n\times m}$, where $R_{ip}=1$ whenever $i\in e_p$. If we do not concern ourselves with the order in which the interactions appear, then hyper-edge exchangeability is a moot point, since the set of interactions is unordered.
\end{remark}

\begin{definition}[Generalized random dot product interaction hypergraph]
When $\omega$ takes the form $\omega((x,x'),(y,y'))=\langle x,y\rangle - \langle x',y'\rangle$, so that $\mathcal{X}=\mathcal{Y}=\RR^{d_1}\times \RR^{d_2}$, we call the interaction hypergraph that arises a \emph{Generalized Random Dot Product Interaction Hypergraph (GRDPH).} If $d_2=0$, this is a \emph{Random Dot Product Interaction Hypergraph (RDPH).}
\end{definition}

\begin{definition}[Stochastic blockmodel interaction hypergraph]
When $\mathcal{X}=[d]$ is a discrete set, we call the generated interaction hypergraph a \emph{Stochastic Blockmodel Interaction Hypergraph (Hyper-SBM).} In this setting, we say that any vertices with latent position $r\in[d]$ belong to class $r$, and denote the number of nodes belonging to each class $r$ by $n_r$. We may also define a \emph{type vector} $\tau_p$ for each interaction $p$, so that its $r$th entry $\tau_{rp}$ is equal to the number of nodes from class $r$ belonging to this interaction. The \emph{type matrix} $\mathcal{T}\in \NN^{d\times m}$ collects the type vectors together as its columns. Given $\mathcal{T}$, the interactions are sampled uniformly without replacement, with $\tau_{rp}$ entries being selected from the nodes belonging to class $r$, independently across classes and interactions. Note, however, that the presence or absence of vertices from the same class within any particular interaction are not independent events.
\end{definition}

It is occasionally more convenient to refer to the \emph{basic type vectors} $\bar{\tau}_p$, which binarize the ordinary type vectors to indicate the membership of any nodes from each class in the interaction $p$. That is, $\bar{\tau}_{rp}=1$ if $\tau_{rp}>0$, and otherwise $\bar{\tau}_{rp}=0$. The corresponding matrix is denoted $\bar{\mathcal{T}}$.

In what follows, we will consider the case of Hyper-SBMs, and consider the spectral analysis of their incidence matrices. Given the type vectors, we may write the mean matrix $\Gamma=\EE[R|\mathcal{T}]$ as $ZB\mathcal{T}$, where $Z\in \{0,1\}^{n\times d}$ gives the class memberships for each of the $n$ vertices, and $B=\mathrm{diag}(1/n_r)$. Indeed, when vertex $i$ belongs to class $r$, we have
$$ \EE[R_{ip}|\tau_p] = \PP[i \in e_p | \tau_{rp}] = \frac{\tau_{rp}}{n_r}.$$ As such, the mean matrix $\Gamma$ always has rank at most $d$.

We define the SVD of $\Gamma$ as
$$\Gamma=[U|U_{\perp}][S\oplus 0][V|V_{\perp}]^\top,$$ but rather than working with the SVD of $R$ directly, we will instead consider a certain subset of the eigenvectors of $\hol{RR^\top}=RR^\top-\mathrm{diag}(RR^\top)$. In order to define the eigenspace precisely, consider the following lemma:

\begin{lemma}[Eigendecomposition of $\mathcal{H}(\EE(RR^\top))$]
\label{lem:eigenH}
The mean matrix $\mathcal{H}(\EE(RR^\top))$ may be split into its action on orthogonal invariant subspaces as 
$$
\mathcal{H}(\EE(RR^\top))=ZB\Sigma_{\mathcal{T}}BZ^\top-\bigoplus_{r=1}^d \mu_r (I_{n_r}-J_{n_r}/n_r),
$$
where $\Sigma_{\mathcal{T}}=\mathcal{T}\mathcal{T}^\top-\mathrm{diag}\left(\mathcal{T}j\right)$, for $j\in \RR^m$ the vector of all ones, and $\mu_r=\sum_{p=1}^m \frac{\tau_{rp}(\tau_{rp}-1)}{n_r(n_r-1)}.$ 
Thus, $\mathcal{H}(\EE(RR^\top))$ has $d$ eigenvalues equal to those of $\sqrt{B}\Sigma_{\mathcal{T}}\sqrt{B}$ with corresponding eigenvectors contained in $\mathrm{range}(Z)$, and $n-d$ negative eigenvalues $-\mu_1,\ldots, -\mu_d$, with multiplicities $n_1-1,\ldots, n_d-1$. 
\end{lemma}

Let $\tilde{\Lambda}$ be the diagonal matrix of eigenvalues of $ZB\Sigma_{\mathcal{T}}BZ^\top$. Then we may define the eigendecomposition of $\hol{\EE(RR^\top)}$ as
\begin{equation}
\hol{\EE(RR^\top)}=[\tilde U|\tilde U_{\perp}][\tilde \Lambda\oplus\tilde\Lambda_{\perp}][\tilde U|\tilde U_{\perp}]^\top.
\label{eq:tildeu}
\end{equation}
Note that $\mathrm{span}(\tilde{U})\subseteq \mathrm{span}(Z)$.

Let us define 
\begin{equation}
\Delta = \min_{r,s} \left|\lambda_r(\sqrt{B}\Sigma_{\mathcal{T}}\sqrt{B})-(-\mu_s)\right|.
\label{eq:delta}
\end{equation}
This will serve as a measure of the signal strength in the matrix $\mathcal{H}(\EE(RR^\top))$. Under the very weak assumption $\Delta>0$, we may distinguish the eigenvalues of $\mathcal{H}(\EE(RR^\top))$ belonging to the term $ZB\Sigma_{\mathcal{T}}BZ^\top$ from the values $-\mu_r$, since the multiplicities for the latter are $n_r-1$, which is much, much larger than the $d$ eigenvalues of this first term. In fact, with a stronger assumption on $\Delta$, we can ensure that we can distinguish these eigenvalues even from the noisy matrix $\hol{RR^\top}$. Throughout this work, we will make the following assumption:

\begin{assump}[Balanced communities]
\label{assump:community}
There exists a constant $\tilde c>0$ s.t
$$n_r\geq \tilde cn,\quad r=1,\dots,d$$
\end{assump}

\begin{lemma}
\label{lem:eigtrapping}
Suppose that $m\geq n$, Assumption~\ref{assump:community} holds, and we have the signal strength condition
\begin{equation}
\Delta \geq 21\sqrt{m\log(m)k_{\mathrm{max}}\bar{k}/\tilde{c}}=:3b.
\label{eq:bigdelta}
\end{equation}
Then with overwhelming probability as $m\rightarrow\infty$, ordering eigenvalues using the natural order on the real line,
$$ 
\max_{1\leq i\leq n} |\lambda_i(\hol{RR^\top})-\lambda_i(\hol{\EE(RR^\top)})|\leq b.
$$
In particular,
\begin{itemize}
    \item $\hol{RR^\top}$ will have $n_r-1$ eigenvalues in $-\mu_r\pm b$, for each $1\leq r\leq d,$
    \item The remaining eigenvalues of $\hol{RR^\top}$ will lie in the intervals $I_r:=\lambda_r(ZB\Sigma_{\mathcal{T}}BZ^\top)\pm b,$
    \item If the union of $k$ of the intervals $I_r$ form a maximal connected set $S$ of this form, so that no additional interval $I_{s}$ could be added to this union while remaining connected, then $S$ contains exactly $k$ eigenvalues of $\hol{RR^\top}$.
\end{itemize}
\end{lemma}

In light of this lemma, we may choose the $d$ eigenvalues of $\hol{RR^\top}$ not belonging to any of the intervals $-\mu_r\pm b$ for $1\leq r\leq d$, and arrange these as the diagonal elements of $\Lambda$. Then we define the eigendecomposition of this matrix as
$$
\hol{RR^\top}=[\hat U|\hat U_{\perp}][\hat \Lambda\oplus\hat\Lambda_{\perp}][\hat U|\hat U_{\perp}]^\top.
$$
We will also consider the SVD of $\hat U^\top R$, defined as
$$\hat U^\top R=\hat X\hat S\hat V^\top.$$ As we will soon show, $\hat{U}$ will span approximately the same subspace as $U$ and $\tilde{U}$, which span exactly the same space; $\hat{S}$ will approximate $S$; and $\hat{V}$ will span approximately the same subspace as $V$. While this sort of result might be expected in light of standard results from the statistical analysis of networks, we emphasize that the dependence within interactions mean that the distributional properties of $R$ are quite different from an ordinary adjacency matrix with conditionally independent edge formation, pairwise edges, and no sense in which edge exchangeability holds. This leads to very different arguments in order to arrive at a similar conclusion.

\section{Theoretical Properties}
\label{sec:theory}

A key feature of our theory is that it applies in the settings where $\bar{k}, k_{\mathrm{max}}$ are both of constant order, or when both of these quantities grow; we also allow $n$ to remain bounded or to grow with $m$, so long as $m\geq n$. Rather than splitting into several different cases for $\bar{k}, k_{\mathrm{max}}$, and $n$, we impose the following condition, which jointly controls all of the relevant quantities.

\begin{assump}[Sufficient Information]
\label{assump:information}
Consider the quantity $\Delta$ defined in Equation~\ref{eq:delta}, and let $\kappa=\sigma_1(\Gamma)/\sigma_d(\Gamma)$. As $m\rightarrow\infty$, we suppose the following condition on the growth of the signal strength:
$$\Delta^2 \gg m\log^2(m)\kappa^{14}k_{\mathrm{max}}^6 \bar{k} d^3.$$
\end{assump}

The following two examples demonstrate the requirements of this condition in more concrete settings. From the perspective of networks, this condition may seem very strong, since the number of observed edges can never exceed $n^2/2$. We first point out that in an edge-exchangeable setting, this restriction no longer applies, but there are a few additional subtleties to point out. The role of $k_{\mathrm{max}}$ is of particular interest, since larger interactions contain both more signal as well as additional variance, so the true relationship between the signal strength is more complicated than this sufficient condition would suggest: see Figure~\ref{fig:ARItables} and the supplementary figures in Appendix~\ref{sec:additional} for more on this point. As $k_{\mathrm{max}}$ gets larger, this forces $n$ to grow, and also significantly increases the number of possible clusters of the interaction types, meaning that the clustering problem becomes dramatically more challenging for larger $k_{\mathrm{max}}$. If we were to consider a weaker notion of clustering performance, the corresponding signal strength condition should also become much more lenient. 

\begin{example}
Suppose that $d=2$, $n_1=n_2=n/2$, and suppose that there are $\alpha m$ pure interactions involving $k$ type 1 vertices, $\alpha m$ pure interactions involving $k$ type 2 nodes, and $(1-2\alpha)m$ mixed interactions involving $k$ nodes of each type. Then $\mu_1=\mu_2=(1-\alpha) m \frac{2k(k-1)}{n(n-2)}$, and $\sqrt{B}\Sigma_{\mathcal{T}}\sqrt{B}=\frac{2m}{n}\begin{bmatrix}(1-\alpha)k(k-1)& (1-2\alpha)k^2\\ (1-2\alpha)k^2&(1-\alpha)k(k-1) \end{bmatrix},$ which has eigenvalues $(2m/n)\{(2-3\alpha)k^2-(1-\alpha)k,\alpha k^2-(1-\alpha) k\}.$ When $\alpha\in (1/(k+1),(2k-1)/(3k-1)),$ this matrix is positive definite, corresponding to an assortative setting. For this example, $\kappa=(2-3\alpha)/\alpha.$ In particular, for fixed $k$, we have $\Delta^2 \sim (m/n)^2$.
\end{example}

\begin{example}
Suppose that $d=2$, $n_1=n_2=n/2$, and that all interactions have size $k_p=3k$, with half having $\tau_{1p}=2k,\tau_{2p}=k$, and the other half having $\tau_{1p}=k, \tau_{2p}=2k$. Then $\mu_1=\mu_2=\frac{2m [2k(2k-1)+k(k-1)]}{n(n-2)}=\frac{2m(5k^2-3k)}{n(n-2)}$. On the other hand, $\sqrt{B}\Sigma_{\mathcal{T}}\sqrt{B} =\frac{m}{n}\begin{bmatrix}5k^2-3k&4k^2\\ 4k^2&5k^2-3k\end{bmatrix},$ which has eigenvalues $(m/n)\left\{3k(3k-1),k(k-3)\right\}.$ We can check that $\kappa=9$. When $k=1,2,$ or $3$, this corresponds to a disassortative setting, but the eigenvalues are still well-separated from the values $-\mu_1, -\mu_2$. For any fixed $k$ and large enough $m,n$, we have $\Delta^2 \sim (m/n)^2$.
\end{example}

In order to obtain a perfect clustering of the interactions, we will show that the estimates of the interaction latent positions given by the rows of $\hat{V}\hat{S}$ are uniformly close to their theoretical counterparts, which perfectly reveal the communities of the type vectors since $SV_p=S V_q$ if and only if $\tau_p=\tau_q$. However, the number of clusters can be very large when $k_{\mathrm{max}}$ is large, which makes the clustering problem challenging. To remedy this, we consider applying hierarchical clustering with complete linkage \citep{johnson1967hierarchical}. This algorithm is provided in Algorithm~\ref{alg:hc}. Analysis of this procedure reveals that at some point along the sequence of cluster merges, the exact clusters of the type vectors are recovered. If one knew the true number of distinct type vectors, they could then recover the true clustering, but the proof of this theorem also shows that there must be a significant increase in the linkage function just after the true number of clusters has been passed, suggesting that in practice, this number could be recovered.

\begin{theorem}[Perfect clustering in nonuniform hypergraphs]
\label{thm:scaledv}
Suppose Assumptions~\ref{assump:community} and \ref{assump:information} hold. Then with overwhelming probability as $m\rightarrow\infty$, 
$$
\|\hat{V}\hat{S}-VSW\|_{2\rightarrow\infty} \lesssim \frac{\sqrt{m}\log(m)\kappa^7 k_{\mathrm{max}}^3\sqrt{\bar{k}} d^{3/2}}{\sqrt{n}\Delta},
$$
so $\sqrt{n}\|\hat{V}\hat{S}-VSW\|_{2\rightarrow\infty}\rightarrow0$ as $m\rightarrow\infty$.

On the high-probability event where $\sqrt{n}\|\hat{V}\hat{S}-VSW\|_{2\rightarrow\infty}\leq 1/(10\tilde{c})$, if we consider the sequence of clusterings returned by hierarchical clustering with complete linkage applied to the rows of $\hat{V}\hat{S}$, one of these clusterings perfectly recovers the true communities of the type vectors $\{\tau_p\}$.
\end{theorem}

\begin{algorithm}[t]
\caption{Hierarchical clustering with complete linkage}\label{alg:hc}
\begin{algorithmic}
\Require Points $x_1,\ldots,x_n\in \RR^d$
\Ensure Nested clusterings $\{\{C_j^{(i)}\}_{j=1}^{n_i}\}_{i=0}^{n-1}$
\State $C_j^{(0)}\gets \{j\}$ for $1\leq j\leq n$, $n_0\gets n$.
\For{$i=1,\ldots,n-1$}
    \State For $1\leq k<\ell\leq n_{i-1}$, compute
        $$d_{k,\ell}^{(i-1)}=\max\left\{\|x_p-x_q\|:p\in C_k^{(i-1)}, q\in C_\ell^{(i-1)}\right\}.$$
    \State For the pair $(k^*,\ell^*)$ minimizing $d_{k,\ell}^{(i-1)}$, set $C_1^{(i)}\gets C_{k^*}^{(i-1)}\cup C_{\ell^*}^{(i-1)}.$
    \State Set the remaining clusters so that $\{C_j^{(i)}\}_{j=2}^{n_{i-1}-1}=\{C_j^{(i-1)}: 1\leq j\leq n_{i-1}, j\neq k^*,\ell^*\}.$
    \State $n_i \gets n-i$
\EndFor
\end{algorithmic}
\end{algorithm}

\section{Simulations}
\label{sec:simulation}

We verify our theoretical results in simulated interaction hypergraphs. We generated hypergraphs with $m=999\times \{3^0,3^1,3^2,3^3,3^4,3^5\}$ interactions and $n=10\times \{2^0,2^1,2^2,2^3,2^4,2^5\}$ nodes, for a total of $6\times 6=36$ combinations. In all experiments, we suppose that nodes belong to $d=2$ classes, so there are $s=3$ basic interaction types, i.e. $2$ pure and $1$ mixed basic types of interactions. For simplicity, we set $n_r=n/2, r=1,2$ and the basic type vectors as
$$\bar\tau_{p}=
\begin{cases}
(1,0)&\text{if }0\leq p\leq m/3,\\
(0,1)&\text{if }m/3\leq p\leq 2m/3,\\
(1,1)&\text{if }2m/3\leq p\leq m.
\end{cases}$$
In other words, nodes and interactions are evenly allocated to each class and basic type. For each interaction $p=1,\dots,m$, we sample $k_p$ as follows: For pure interactions, $k_p\geq 2=:k_{\mathrm{min}}$, while for mixed interactions, $k_p \geq \sum_r \bar \tau_{rp}$, where this is the basic type vector for the interaction $p$, so we can sample at least one node from each node class represented in the interaction. Since $d=2$, this means that all interactions have $k_p\geq 2$. The number of nodes involved in the $p$-th interaction follows the distribution $k_{\mathrm{min}}+\mathrm{Binomial}(k_{\mathrm{max}}-k_{\mathrm{min}},\alpha=0.4)$. In particular, this means that $\EE[k_p]=0.6 k_{\mathrm{min}}+0.4 k_{\mathrm{max}}.$ In settings where mixed interactions $k_p\geq k_{\mathrm{min}}^{\mathrm{mixed}}\neq k_{\mathrm{min}}^{\mathrm{pure}}$, one could consider different distributions of the interaction sizes for the pure and mixed cases, but this is not necessary in the present setting.

We consider two settings, one with growing $k_{\mathrm{max}}$ and one with fixed $k_{\mathrm{max}}$. For growing $k_{\mathrm{max}}$, the maximum interaction size is set to be $k_{\mathrm{max}}=n/d$, so $\EE[k_p]= 1.2+0.2 n$. For fixed $k_{\mathrm{max}}$, we set $k_{\mathrm{max}}=5$, so $\EE[k_p]=3.2$. An example of the right singular vectors for the case $m=999, n=40,$ and growing $k_{\mathrm{max}}$ (so $k_{\mathrm{max}}=20$) can be seen in Figure~\ref{fig:IH-sim-clus}.

\begin{figure}[h]
\centering
\begin{subfigure}{0.49\textwidth}
    \includegraphics[width=\textwidth]{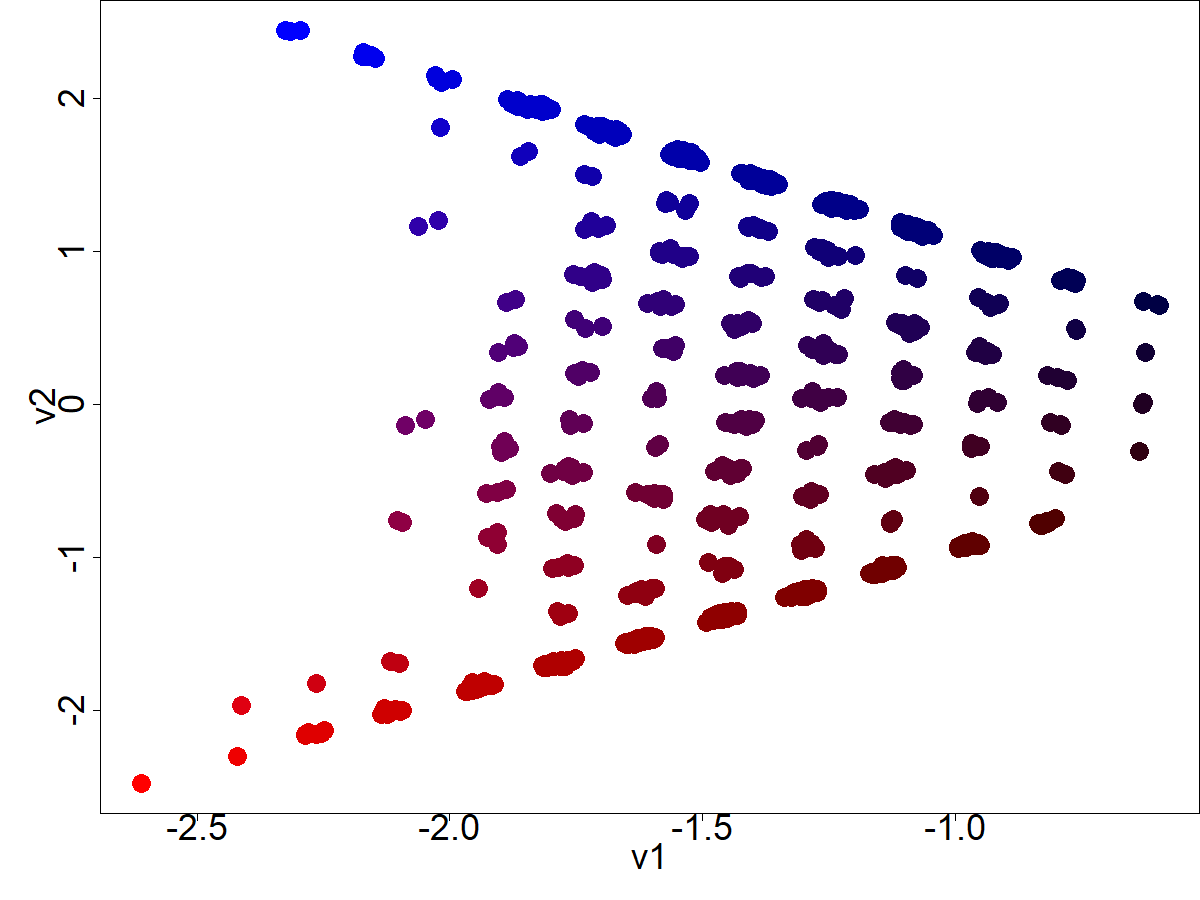}
    \caption{Scatter plot of first two columns of $\hat{V}\hat{S}$.}
\end{subfigure}
\hfill
\begin{subfigure}{0.49\textwidth}
    \includegraphics[width=\textwidth]{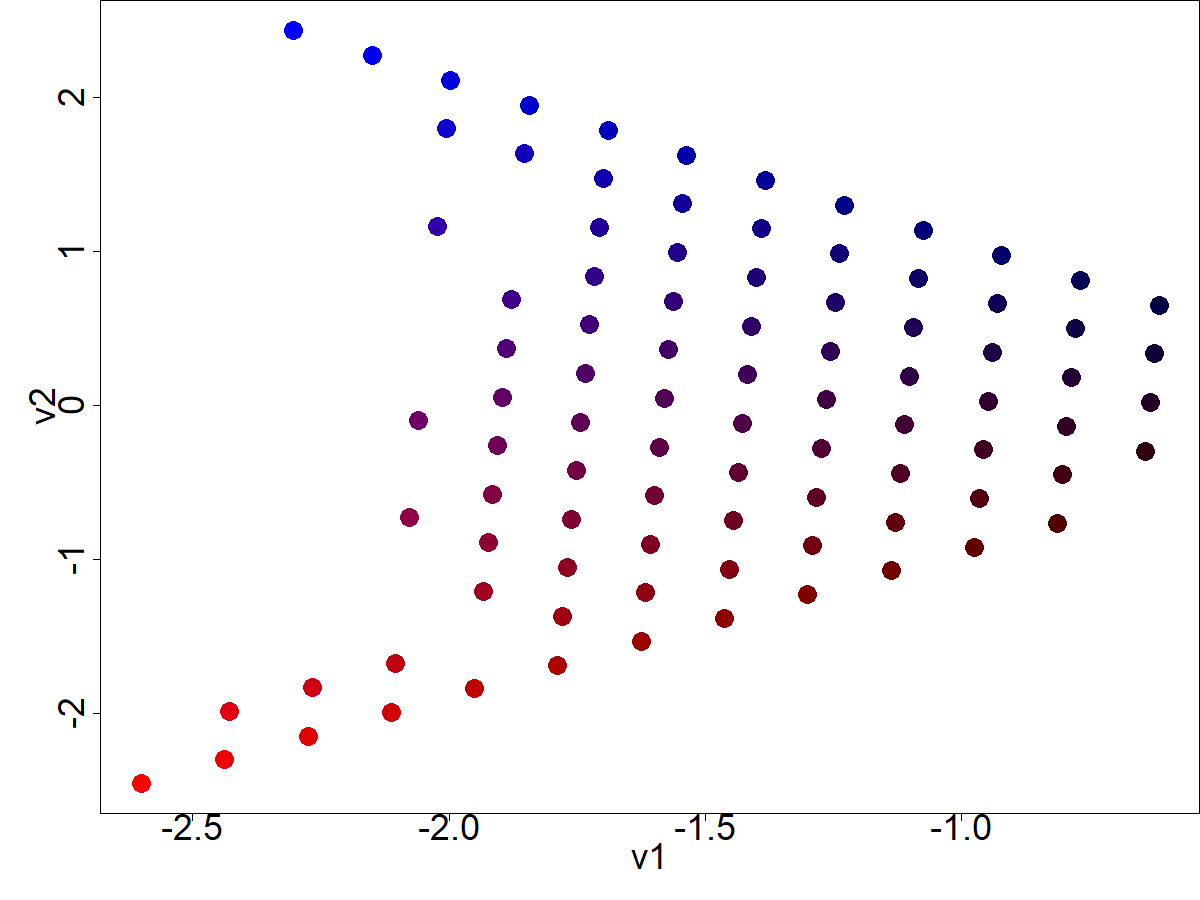}
    \caption{Scatter plot of first two columns of $VS$.}
\end{subfigure}
\caption{Scatter plots of estimated and theoretical right singular vectors with $m=999,n=40, k_{\mathrm{max}}=20$, colored by the type vectors of the interactions.}
\label{fig:IH-sim-clus}
\end{figure}

We visualize the decrease in the quantity $\|\hat{V}\hat{S}-VSW\|_{2\rightarrow\infty}$ as $m$ grows in Figure~\ref{fig:vhatminusv}. Since we show these results in log-log scale, monomials in $m$ appear as linear functions, enabling easier comparison of the observed rates with our empirical ones. 

\begin{figure}[h]
\centering
\begin{subfigure}{0.49\textwidth}
\includegraphics[width=\textwidth]{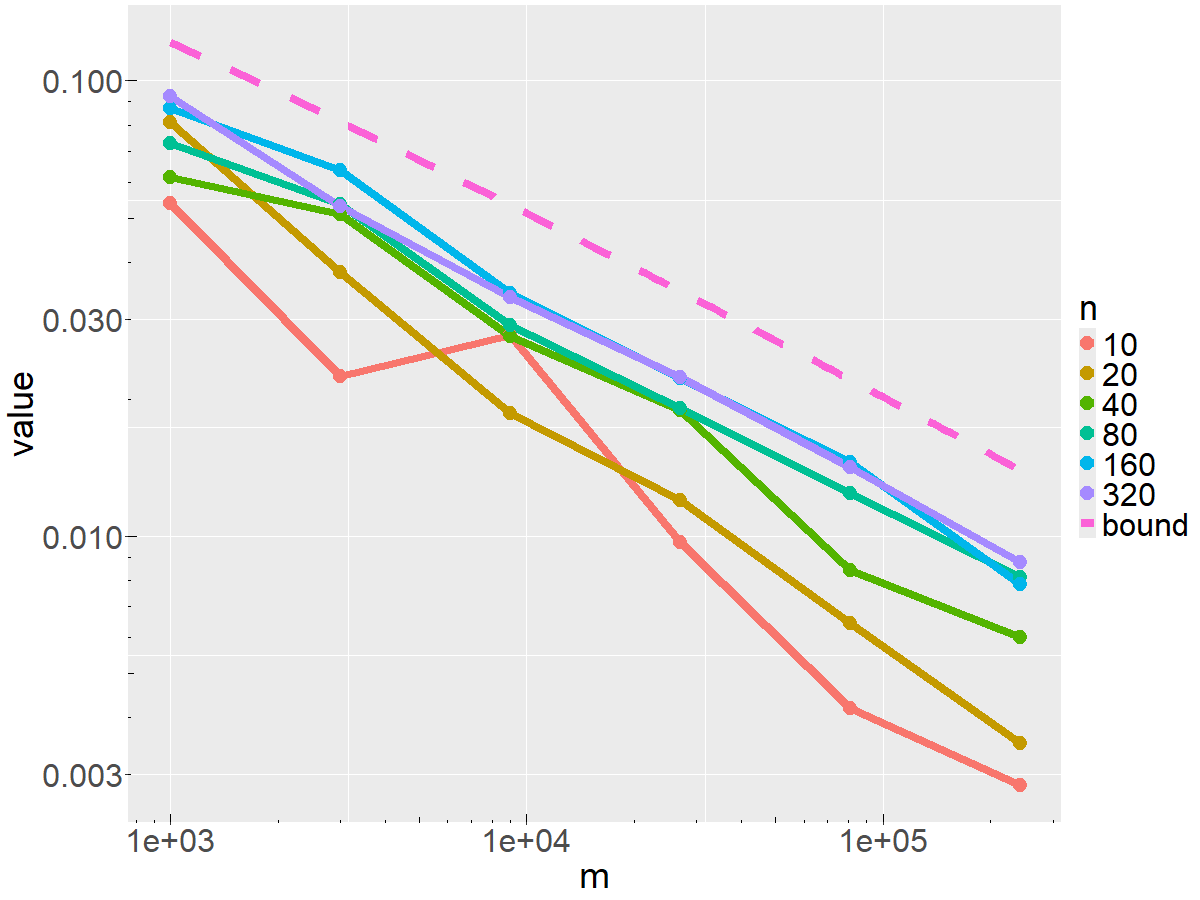} 
\caption{Growing $k_{\mathrm{max}}=n/2$.}
\end{subfigure}
\hfill
\begin{subfigure}{0.49\textwidth}
\includegraphics[width=\textwidth]{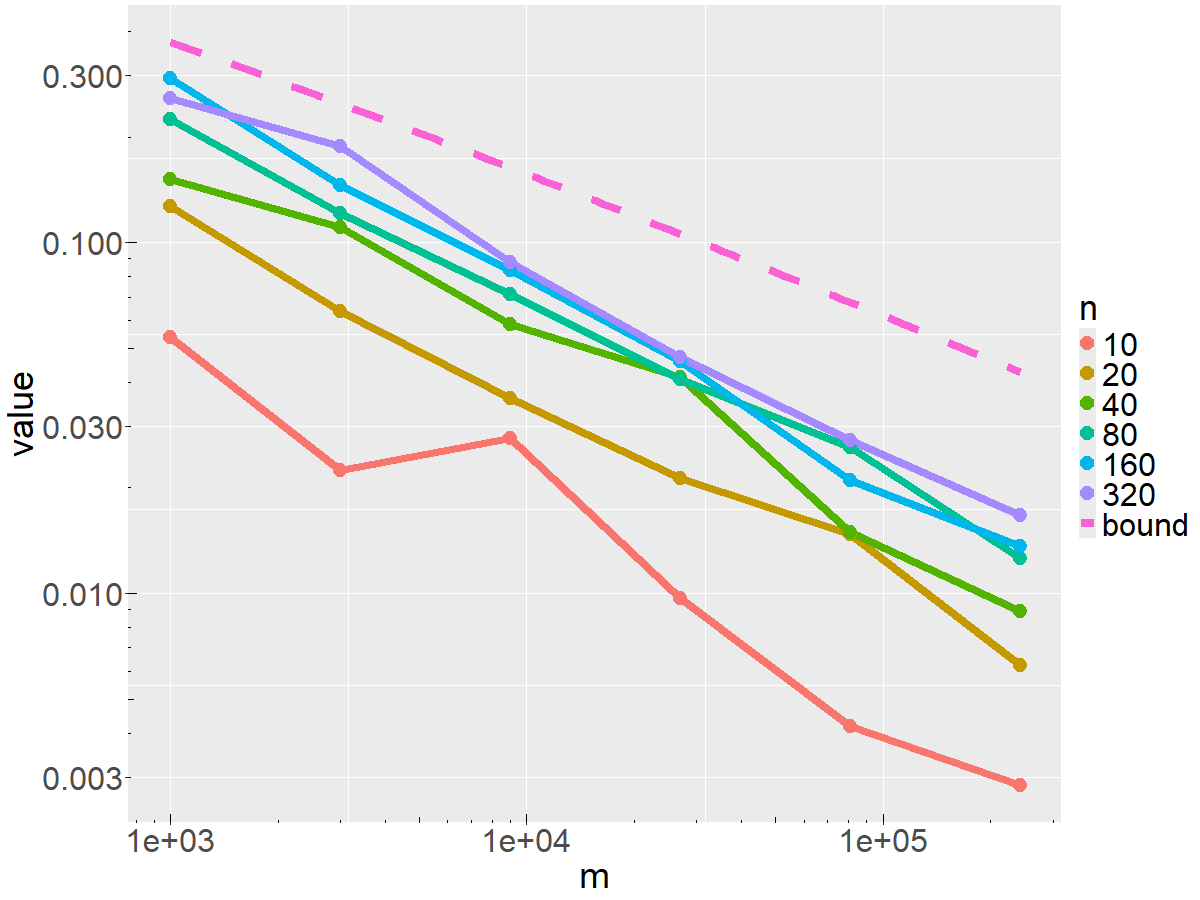}
\caption{Fixed $k_{\mathrm{max}}=5$.}
\end{subfigure}
\caption{Decrease of $\|\hat{V}\hat{S}-VSW\|_{2\rightarrow\infty}$ as $m$ grows, for various values of $n$. Both panels are shown in log-log scale. The reference curves are given by the bound in Theorem~\ref{thm:scaledv}, with intercepts chosen to place them above the empirical curves, taking into account any unknown constants.}
\label{fig:vhatminusv}
\end{figure}

In Figure~\ref{fig:ARItables}, we show the change in ARI as $m$ and $n$ vary, for the cases of growing $k_{\mathrm{max}}$ or fixed, as before. For clustering, since the number of clusters is very large, $K$-means clustering is unstable without very careful initialization or many restarts. As such, in both our theoretical results and in these simulations, we use hierarchical clustering on the estimated interaction latent positions with complete linkage. We find that the true clusters are recovered remarkably well across the range of parameters that we consider, in keeping with our theoretical results. In particular, the worst behavior in the clustering performance is observed in the case of fixed $k_{\mathrm{max}}$, large $n$, and small $m$. This is most likely caused by the reduced information provided by smaller interactions, which is not yet compensated by the reduction in variance. For large $n$, as $m$ grows, we see that the large variance in the case of growing $k_{\mathrm{max}}$ still causes a small number of errors in the clustering, but the additional signal strength dominates the variance in the case of fixed $k_{\mathrm{max}}$.

\begin{figure}
\centering
\begin{subfigure}{0.8\textwidth}
\centering
\begin{tabular}{c|c|c|c|c|c|c|}
$n$, $m$ & 999 & 2997 & 8991 & 26973 & 80919 & 242757\\
\hline
10 & 1.000 & 1.000 & 1.000 & 1.000 & 1.000 & 1.000\\
20 & 1.000 & 1.000 & 1.000 & 1.000 & 1.000 & 1.000\\
40 & 0.997 & 0.994 & 0.994 & 0.994 & 0.994 & 0.994\\
80 & 0.932 & 0.997 & 0.996 & 0.996 & 0.996 & 0.996\\
160 & 0.819 & 0.944 & 0.996 & 0.997 & 0.997 & 0.997\\
320 & 0.607 & 0.741 & 0.978 & 0.998 & 0.998 & 0.998\\
\hline
\end{tabular}
\caption{Growing $k_{\mathrm{max}}=n/2$.}
\end{subfigure}
\vskip6pt
\begin{subfigure}{0.8\textwidth}
\centering
\begin{tabular}{c|c|c|c|c|c|c|}
$n$, $m$ & 999 & 2997 & 8991 & 26973 & 80919 & 242757\\
\hline
10 & 1.000 & 1.000 & 1.000 & 1.000 & 1.000 & 1.000\\
20 & 1.000 & 1.000 & 1.000 & 1.000 & 1.000 & 1.000\\
40 & 0.972 & 0.964 & 1.000 & 1.000 & 1.000 & 1.000\\
80 & 0.745 & 0.963 & 1.000 & 1.000 & 1.000 & 1.000\\
160 & 0.407 & 0.702 & 0.944 & 1.000 & 1.000 & 1.000\\
320 & 0.361 & 0.506 & 0.812 & 1.000 & 1.000 & 1.000\\
\hline
\end{tabular}
\caption{Fixed $k_{\mathrm{max}}=5.$}
\end{subfigure}
\caption{Comparison of clustering from estimated interaction latent positions with the true type vectors.}
\label{fig:ARItables}
\end{figure}

% \begin{figure}[ht]
% \begin{tabular}{cc}
%   \includegraphics[width=0.5\textwidth]{images/sim_R_Gamma.png} &   \includegraphics[width=0.5\textwidth]{images/sim_RRT_GammaGammaT.png} \\
% (a)$\norm{R-\Gamma}$ & (b) $\norm{RR^\top-\Gamma\Gamma^\top}$ \\
%  \includegraphics[width=0.5\textwidth]{images/sim_SaWUa_WVaShat.png} &   \includegraphics[width=0.5\textwidth]{images/sim_Sa12WUa_WVaShat12_norm.png} \\
% (c)$\norm{\mathcal S\mathcal W^*_U-\mathcal W^*_V\hat{S}}$ & (d)  $\norm{\mathcal S^{1/2}\mathcal W^*_U-\mathcal W^*_V\hat{S}^{1/2}}$ \\
% \includegraphics[width=0.5\textwidth]{images/sim_Sa_12WVa_WUaShat_12_norm.png} &   \includegraphics[width=0.5\textwidth]{images/sim_VhatShat12_VaSa12WVa_2Inf_manuscript.png} \\
% (e)$\norm{\mathcal S^{-1/2}W^*_V-\mathcal W^*_U\hat{S}^{-1/2}}$ & (f)  $\norm{\hat{V}\hat S^{1/2}-\mathcal V\mathcal S^{1/2}\mathcal W^*_V}_{2\rightarrow\infty}$ \\
% \end{tabular}
% \caption{\label{fig:IH-sim}Simulation results under different $m,n$ }
% \end{figure}

\section{Real-World Networks}
\label{sec:realdata}

We apply our approach to interaction hypergraphs obtained from email communications and brain synapses in larval \emph{Drosophila}. The Enron email network \citep{amburg2021planted} records the email communication between employees of the Enron company and includes $m=15653$ emails/ interactions and $n=4423$ email users/ nodes. Each email is considered an interaction where the email users involved, either senders or receivers, are considered the nodes comprising the interaction.

We apply our method to obtain a three-dimensional representation of the email users in the Enron dataset. First, notable nodes that appear at the extremes of the scatter plots of $\hat V$ are identified and found to form clusters. In Figure \ref{Enron-v-3d}, we show selections of the estimated vectors, colored by the numbers of important members from each of the three groups. The colors of the points are encoded by the standard red-green-blue (RGB), where red, green and blue represent the involvement of nodes from clusters 1, 2 and 3 in an interaction, respectively. While we show two particular sets of singular vectors, in practice, all $5$ or more of the singular vectors can be used for any downstream inference procedure.

\begin{figure}[t]
\centering
\begin{subfigure}{0.49\textwidth}
    \includegraphics[width=\textwidth]{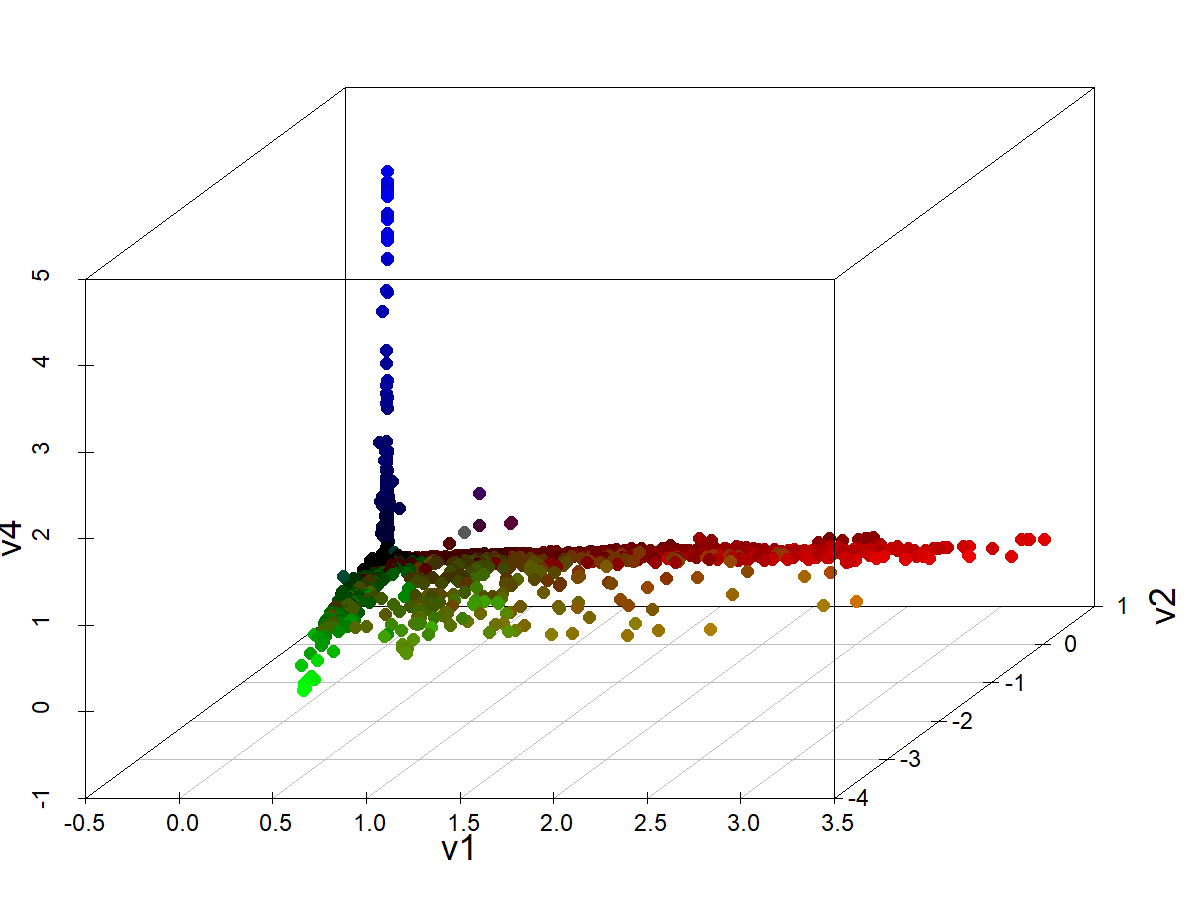}
    \caption{Dimensions $(\hat{V}\hat{S})_1, (\hat{V}\hat{S})_2, (\hat{V}\hat{S})_4$}
\end{subfigure}
\hfill
\begin{subfigure}{0.49\textwidth}
    \includegraphics[width=\textwidth]{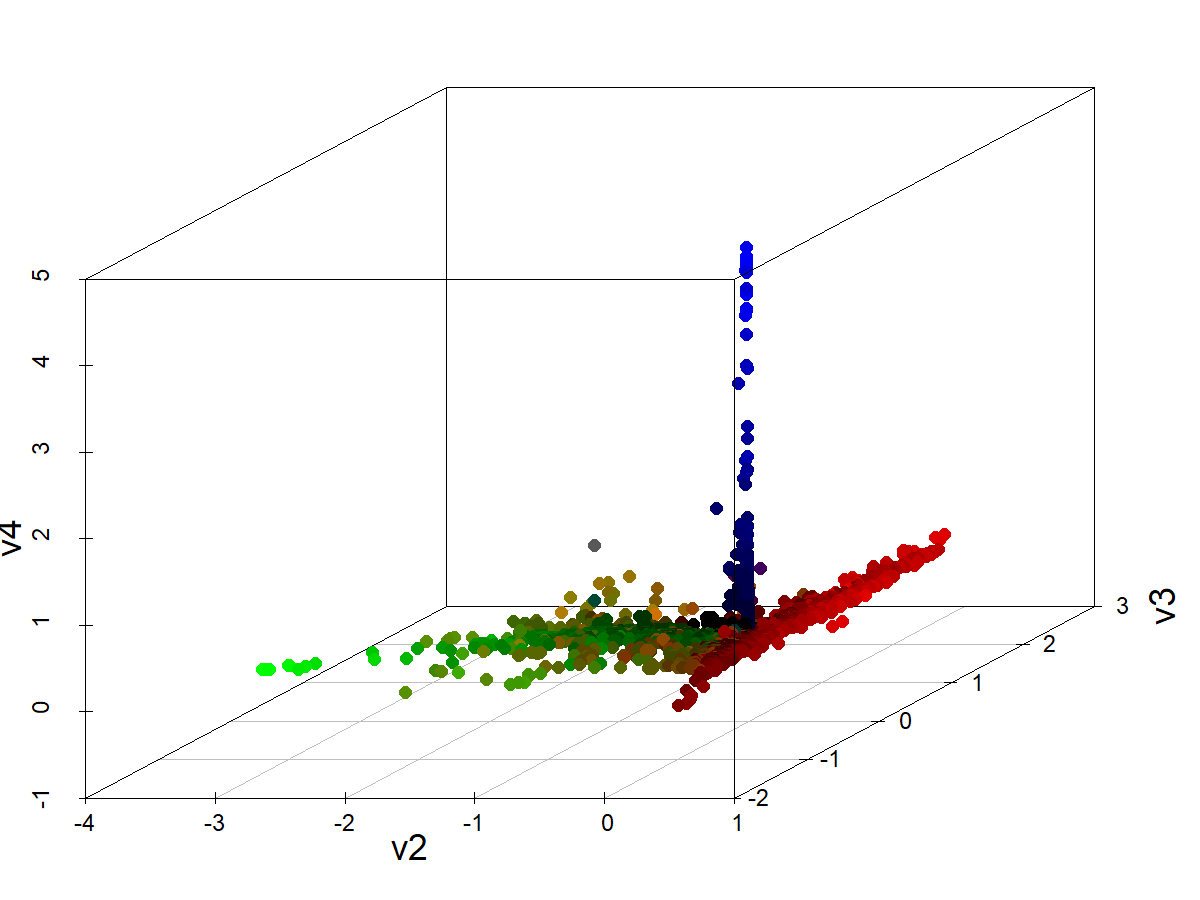}
    \caption{Dimensions $(\hat{V}\hat{S})_2, (\hat{V}\hat{S})_3, (\hat{V}\hat{S})_4$}
\end{subfigure}
\caption{Visualization of the interaction latent position embeddings for the Enron dataset. A small set of influential users are labeled red, blue, and green, and emails are colored by the number of these influential users included.}
\label{Enron-v-3d}
\end{figure}

The Drosophila brain network \citep{winding2023connectome} records the brain connectome of the Drosophila larva. After selecting a spatial region as shown in the left panel of Figure~\ref{Droso-v-3d}, $m=3463$  synapses/ interactions and $n=484$ neurons/ nodes are included. Each synapse is considered an interaction, and the neurons involved, both presynaptic and postsynaptic, are considered the nodes comprising the interaction. We then plot the estimated interaction latent positions from this selected set of interactions, and color the points by the participation of neurons belonging to the left or right brain hemisphere. This reveals an interesting core-periphery structure, with branches of the interaction latent position plot corresponding to groups of interactions mostly belonging to the left brain, mostly belonging to the right brain, and well-mixed interactions, all of which appear to be playing distinct roles in this brain region.

\begin{figure}[t]
\centering
\begin{subfigure}{0.49\textwidth}
    \includegraphics[width=\textwidth]{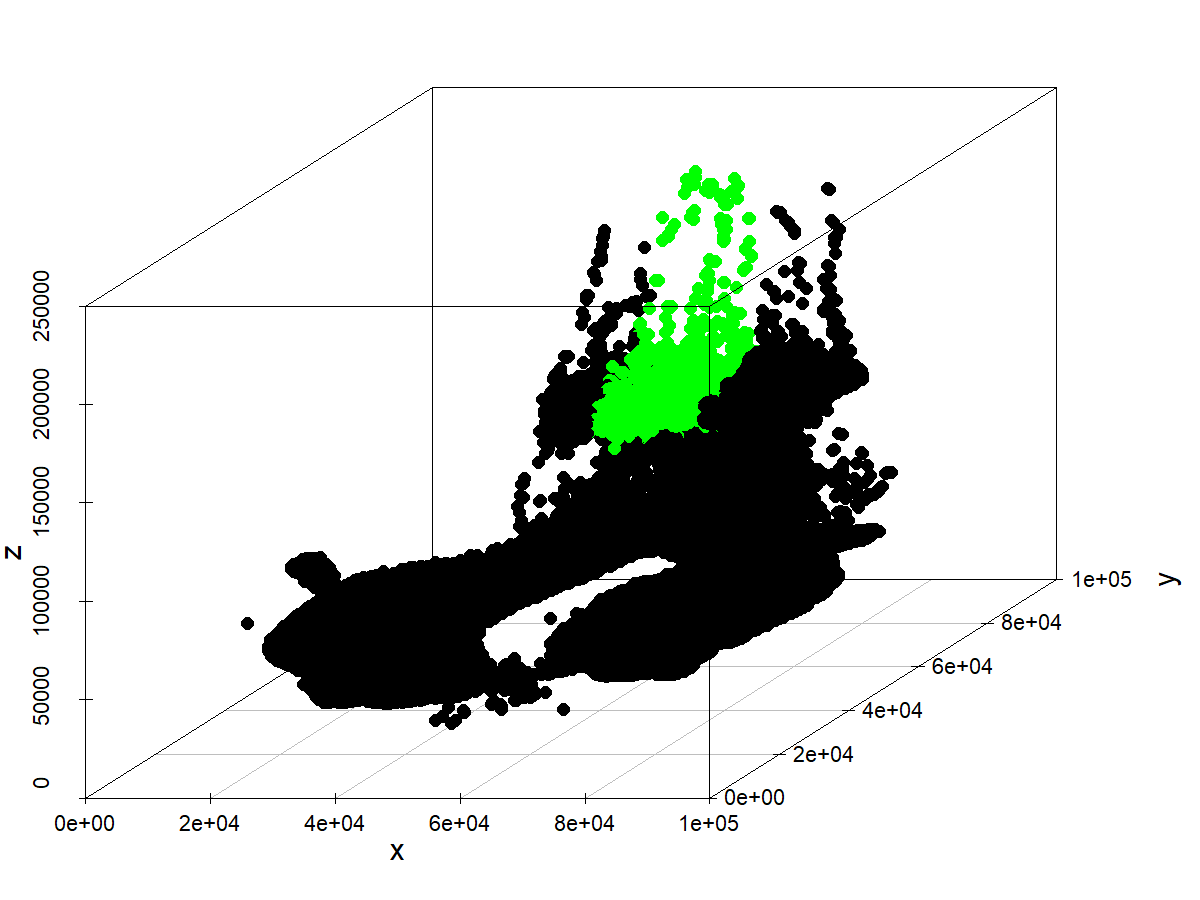}
    \caption{Spatial locations of all synapses. Synapses falling into the selected region are highlighted in green.}
\end{subfigure}
\hfill
\begin{subfigure}{0.49\textwidth}
    \includegraphics[width=\textwidth]{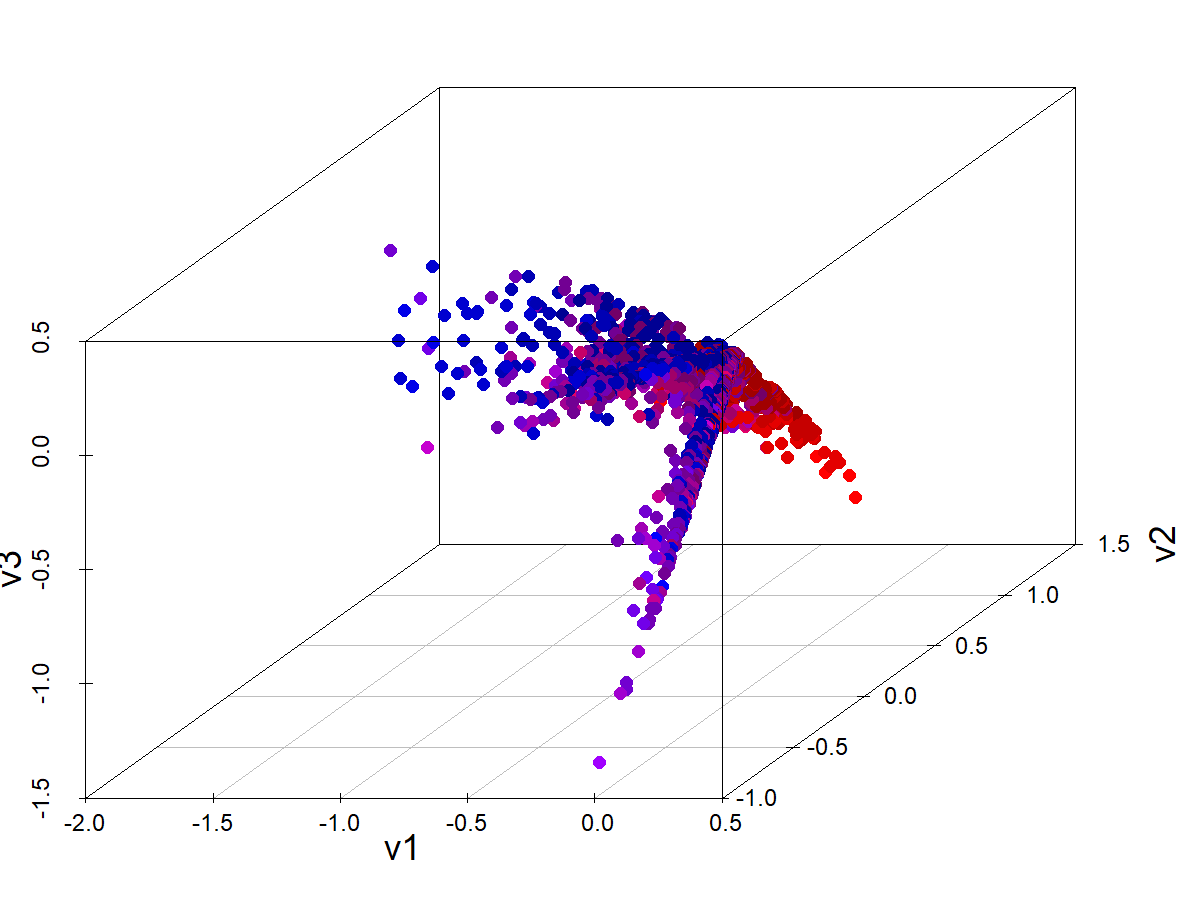}
    \caption{Dimensions $(\hat{V}\hat{S})_1, (\hat{V}\hat{S})_2, (\hat{V}\hat{S)}_3$ of our estimated interaction latent position embeddings.}
\end{subfigure}
\caption{Visualizations of the \emph{Drosophila} larva brain dataset. For the right panel, blue interactions contain only neurons from the left brain hemisphere, while red interactions contain only neurons from the right brain hemisphere. Mixed interactions are shaded by the contributions of left and right hemisphere neurons.}
\label{Droso-v-3d}
\end{figure}

\section{Discussion}
\label{sec:discussion}
We have considered statistical models for random hypergraphs whose sampling units are the \emph{interactions}, rather than the entities in the network. We argue that these models better reflect the sampling mechanisms used for observing real-world hypergraphs, while retaining statistical interpretability and theoretical tractability. In a blockmodel-like setting, we have shown under a mild signal strength condition that we can identify a subspace that captures node communities. Using this projection leads to latent position estimates for the interactions, revealing more detailed structure about the hypergraph than the weighted adjacency matrix alone. When the signal strength is sufficient, these interaction latent position estimates concentrate tightly around the type vectors (up to an orthogonal transformation). We then demonstrated these results through simulation, and tested the methodology with real data from the Enron email network and \emph{Drosophila} larva brain. Future work can extend these results to more general settings, for example, the (G)RDPH or LPH models. Refined methodology may improve the signal strength condition that is required for tight concentration of the interaction latent positions. Finally, incorporating interaction covariate information with these latent positions may enable stronger inference about these networks.

\bibliography{bibliography}

\begin{thebibliography}{25}
\providecommand{\natexlab}[1]{#1}
\providecommand{\url}[1]{\texttt{#1}}
\expandafter\ifx\csname urlstyle\endcsname\relax
  \providecommand{\doi}[1]{doi: #1}\else
  \providecommand{\doi}{doi: \begingroup \urlstyle{rm}\Url}\fi

\bibitem[Agterberg and Zhang(2022)]{agterberg2022estimating}
Joshua Agterberg and Anru Zhang.
\newblock Estimating higher-order mixed memberships via the $\ell_{2,\infty}$
  tensor perturbation bound.
\newblock \emph{arXiv preprint arXiv:2212.08642}, 2022.

\bibitem[Amburg et~al.(2021)Amburg, Kleinberg, and Benson]{amburg2021planted}
Ilya Amburg, Jon Kleinberg, and Austin~R Benson.
\newblock Planted hitting set recovery in hypergraphs.
\newblock \emph{Journal of Physics: Complexity}, 2\penalty0 (3):\penalty0
  035004, 2021.

\bibitem[Banerjee(2021)]{banerjee2021spectrum}
Anirban Banerjee.
\newblock On the spectrum of hypergraphs.
\newblock \emph{Linear algebra and its applications}, 614:\penalty0 82--110,
  2021.

\bibitem[Boccaletti et~al.(2023)Boccaletti, De~Lellis, Del~Genio,
  Alfaro-Bittner, Criado, Jalan, and Romance]{boccaletti2023structure}
Stefano Boccaletti, Pietro De~Lellis, CI~Del~Genio, Karin Alfaro-Bittner,
  Regino Criado, Sarika Jalan, and Miguel Romance.
\newblock The structure and dynamics of networks with higher order
  interactions.
\newblock \emph{Physics Reports}, 1018:\penalty0 1--64, 2023.

\bibitem[Cardoso(2021)]{cardoso2021principal}
Kau{\^e} Cardoso.
\newblock Principal eigenvector of the signless {L}aplacian matrix.
\newblock \emph{Computational and Applied Mathematics}, 40\penalty0
  (2):\penalty0 50, 2021.

\bibitem[Cardoso and Trevisan(2021)]{cardoso2021general}
Kau{\^e} Cardoso and Vilmar Trevisan.
\newblock Principal eigenvectors of general hypergraphs.
\newblock \emph{Linear and Multilinear Algebra}, 69\penalty0 (14):\penalty0
  2641--2656, 2021.

\bibitem[Chodrow et~al.(2023)Chodrow, Eikmeier, and
  Haddock]{chodrow2023nonbacktracking}
Philip Chodrow, Nicole Eikmeier, and Jamie Haddock.
\newblock Nonbacktracking spectral clustering of nonuniform hypergraphs.
\newblock \emph{SIAM Journal on Mathematics of Data Science}, 5\penalty0
  (2):\penalty0 251--279, 2023.

\bibitem[Cooper and Dutle(2012)]{cooper2012spectra}
Joshua Cooper and Aaron Dutle.
\newblock Spectra of uniform hypergraphs.
\newblock \emph{Linear Algebra and its applications}, 436\penalty0
  (9):\penalty0 3268--3292, 2012.

\bibitem[Crane and Dempsey(2018)]{crane2018edge}
Harry Crane and Walter Dempsey.
\newblock Edge exchangeable models for interaction networks.
\newblock \emph{Journal of the American Statistical Association}, 113\penalty0
  (523):\penalty0 1311--1326, 2018.

\bibitem[Gu et~al.(2022)Gu, Lai, and Song]{gu2022hamiltonian}
Xiaofeng Gu, Hong-Jian Lai, and Sulin Song.
\newblock On {H}amiltonian line graphs of hypergraphs.
\newblock \emph{Journal of Graph Theory}, 100\penalty0 (3):\penalty0 489--503,
  2022.

\bibitem[Han et~al.(2022)Han, Luo, Wang, and Zhang]{han2022exact}
Rungang Han, Yuetian Luo, Miaoyan Wang, and Anru~R Zhang.
\newblock Exact clustering in tensor block model: Statistical optimality and
  computational limit.
\newblock \emph{Journal of the Royal Statistical Society Series B: Statistical
  Methodology}, 84\penalty0 (5):\penalty0 1666--1698, 2022.

\bibitem[Hoff et~al.(2002)Hoff, Raftery, and Handcock]{hoff2002latent}
Peter~D Hoff, Adrian~E Raftery, and Mark~S Handcock.
\newblock Latent space approaches to social network analysis.
\newblock \emph{Journal of the American Statistical association}, 97\penalty0
  (460):\penalty0 1090--1098, 2002.

\bibitem[Horn and Johnson(2012)]{horn2012matrix}
Roger~A Horn and Charles~R Johnson.
\newblock \emph{Matrix Analysis}.
\newblock Cambridge university press, 2012.

\bibitem[Hu and Qi(2012)]{hu2012algebraic}
Shenglong Hu and Liqun Qi.
\newblock Algebraic connectivity of an even uniform hypergraph.
\newblock \emph{Journal of Combinatorial Optimization}, 24\penalty0
  (4):\penalty0 564--579, 2012.

\bibitem[Hu and Qi(2015)]{hu2015laplacian}
Shenglong Hu and Liqun Qi.
\newblock The {L}aplacian of a uniform hypergraph.
\newblock \emph{Journal of Combinatorial Optimization}, 29\penalty0
  (2):\penalty0 331--366, 2015.

\bibitem[Johnson(1967)]{johnson1967hierarchical}
Stephen~C Johnson.
\newblock Hierarchical clustering schemes.
\newblock \emph{Psychometrika}, 32\penalty0 (3):\penalty0 241--254, 1967.

\bibitem[Lin and Zhou(2020)]{lin2020alpha}
Hongying Lin and Bo~Zhou.
\newblock The $\alpha$-spectral radius of general hypergraphs.
\newblock \emph{Applied Mathematics and Computation}, 386:\penalty0 125449,
  2020.

\bibitem[Luo and Zhang(2022)]{luo2022tensor}
Yuetian Luo and Anru~R Zhang.
\newblock Tensor clustering with planted structures: Statistical optimality and
  computational limits.
\newblock \emph{The Annals of Statistics}, 50\penalty0 (1):\penalty0 584--613,
  2022.

\bibitem[Rota~Bul{\`o} and Pelillo(2009)]{rota2009new}
Samuel Rota~Bul{\`o} and Marcello Pelillo.
\newblock New bounds on the clique number of graphs based on spectral
  hypergraph theory.
\newblock In \emph{International Conference on Learning and Intelligent
  Optimization}, pages 45--58. Springer, 2009.

\bibitem[Sun et~al.(2019)Sun, Zhou, and Bu]{sun2019spectral}
Lizhu Sun, Jiang Zhou, and Changjiang Bu.
\newblock Spectral properties of general hypergraphs.
\newblock \emph{Linear Algebra and its Applications}, 561:\penalty0 187--203,
  2019.

\bibitem[Tropp et~al.(2015)]{tropp2015}
Joel~A Tropp et~al.
\newblock An introduction to matrix concentration inequalities.
\newblock \emph{Foundations and Trends{\textregistered} in Machine Learning},
  8\penalty0 (1-2):\penalty0 1--230, 2015.

\bibitem[Winding et~al.(2023)Winding, Pedigo, Barnes, Patsolic, Park,
  Kazimiers, Fushiki, Andrade, Khandelwal, Valdes-Aleman,
  et~al.]{winding2023connectome}
Michael Winding, Benjamin~D Pedigo, Christopher~L Barnes, Heather~G Patsolic,
  Youngser Park, Tom Kazimiers, Akira Fushiki, Ingrid~V Andrade, Avinash
  Khandelwal, Javier Valdes-Aleman, et~al.
\newblock The connectome of an insect brain.
\newblock \emph{Science}, 379\penalty0 (6636):\penalty0 eadd9330, 2023.

\bibitem[Xie and Qi(2015)]{xie2015clique}
Jinshan Xie and Liqun Qi.
\newblock The clique and coclique numbers’ bounds based on the
  {H}-eigenvalues of uniform hypergraphs.
\newblock \emph{Int. J. Numer. Anal. Model}, 12\penalty0 (2):\penalty0
  318--327, 2015.

\bibitem[Xie and Qi(2016)]{xie2016spectral}
Jinshan Xie and Liqun Qi.
\newblock Spectral directed hypergraph theory via tensors.
\newblock \emph{Linear and Multilinear Algebra}, 64\penalty0 (4):\penalty0
  780--794, 2016.

\bibitem[Yu et~al.(2015)Yu, Wang, and Samworth]{yu2015useful}
Yi~Yu, Tengyao Wang, and Richard~J Samworth.
\newblock A useful variant of the {D}avis--{K}ahan theorem for statisticians.
\newblock \emph{Biometrika}, 102\penalty0 (2):\penalty0 315--323, 2015.

\end{thebibliography}

\appendix
\section{Proofs}
\label{sec:proofs}

\begin{lemma}[Spectral norm bound for $R-\Gamma$]
\label{lem:rgamma}
    Let $R\in \{0,1\}^{n\times m}, n\leq m$ be the incidence matrix of a Hyper-SBM, and let $\Gamma = \EE(R|\mathcal T)$. Suppose Assumption~\ref{assump:community} holds, and $m\bar{k} \geq n k_{\mathrm{max}}$. Then with overwhelming probability,
    $$\norm{R-\Gamma}\leq O\left(\sqrt{\frac{m\log(m)\bar k}{n}}\right).$$
\end{lemma}
\begin{proof}
    Let $R_{\ast p}$ and $\Gamma_{\ast p}$ denote the $p$-th columns of $R$ and $\Gamma$, respectively, and let $e_p\in\{0,1\}^m$ be the $p$th standard unit vector. Then we may write
    $$R -\Gamma=\sum_{p=1}^m(R_{\ast p}-\Gamma_{\ast p})e_p^\top,$$
    which is a sum of independent, mean-0 matrices. For any index $p$, we may bound these summands by
    \begin{align}
    \|[R_{\ast p}-\Gamma_{\ast p}]e_p^\top\|^2&=\|R_{\ast p}-\Gamma_{\ast p}\|^2=\sum_{r=1}^d\left[\tau_{rp}(1-\tau_{rp}/n_r)^2+(n_r-\tau_{rp})(\tau_{rp}/n_r)^2\right]\notag\\
    &=\sum_{r=1}^d \frac{\tau_{rp}(n_r-\tau_{rp})}{n_r}\leq k_{\mathrm{max}}.\label{eq:termbound}
    \end{align}
    Denote the matrix variance statistic
    \begin{align}
    v(R -\Gamma)&=\mathrm{max}\left\{\left\|\sum_{p=1}^m\EE\left[(R_{\ast p}-\Gamma_{\ast p})e_p^\top e_p(R_{\ast p}-\Gamma_{\ast p})^\top\right]\right\|,\left\|\sum_{p=1}^m\EE\left[e_p(R_{\ast p}-\Gamma_{\ast p})^\top(R_{\ast p}-\Gamma_{\ast p})e_p^\top\right]\right\|\right\}\notag\\
    &=\mathrm{max}\left\{\left\|\sum_{p=1}^m\EE\left[(R_{\ast p}-\Gamma_{\ast p})(R_{\ast p}-\Gamma_{\ast p})^\top\right]\right\|,\left\|\sum_{p=1}^m\EE\left[e_p(R_{\ast p}-\Gamma_{\ast p})^\top(R_{\ast p}-\Gamma_{\ast p})e_p^\top\right]\right\|\right\}.\label{eq:matrixvariance}
    \end{align}
    The matrix in the left quantity is given by
    \begin{align*}
    \EE\left[(R_{\ast p}-\Gamma_{\ast p})(R_{\ast p}-\Gamma_{\ast p})^\top\right]&=\bigoplus_{r=1}^d \left[\frac{\tau_{rp}(\tau_{rp}-1)}{n_r(n_r-1)}-\frac{\tau_{rp}^2}{n_r^2}\right]J_{n_r}+ \left[\frac{\tau_{rp}}{n_r}-\frac{\tau_{rp}(\tau_{rp}-1)}{n_r(n_r-1)}\right]I_{n_r}\\
    &=\bigoplus_{r=1}^d \frac{\tau_{rp}(n_r-\tau_{rp})}{n_r^2}\left[\frac{n_r}{n_r-1}I_{n_r}-\frac{1}{n_r-1}J_{n_r}\right],
    \end{align*}
    so
    $$
    \left\|\sum_{p=1}^m \EE\left[(R_{\ast p}-\Gamma_{\ast p})(R_{\ast p}-\Gamma_{\ast p})^\top\right]\right\|= \max_r \left[\sum_{p=1}^m\frac{\tau_{rp}(n_r-\tau_{rp})}{n_r^2}\right]\frac{n_r}{n_r-1}=\max_r \sum_{p=1}^m\frac{\tau_{rp}(n_r-\tau_{rp})}{n_r(n_r-1)}.
    $$
    We may further upper bound this quantity by $m\bar{k}/\tilde{c}n$, since $\frac{n_r-\tau_{rp}}{n_r-1}\leq 1$ for all $r,p$; $\max_r \sum_p \tau_{rp}\leq \sum_{r,p} \tau_{rp}=m\bar{k}$; and using the lower bound $n_r\geq \tilde{c}n$ afforded by Assumption~\ref{assump:community}. In sum,
    $$ \left\|\sum_{p=1}^m \EE\left[(R_{\ast p}-\Gamma_{\ast p})(R_{\ast p}-\Gamma_{\ast p})^\top\right]\right\| \leq O(m\bar{k}/n). $$
    To bound the right quantity of Equation~\ref{eq:matrixvariance}, we use the calculation in Equation~\ref{eq:termbound} to obtain
    $$
    \left\|\sum_{p=1}^m\EE\left[e_p(R_{\ast p}-\Gamma_{\ast p})^\top(R_{\ast p}-\Gamma_{\ast p})e_p^\top\right]\right\|=\|\mathrm{diag}(\EE[\norm{R_{\ast p}-\Gamma_{\ast p}}^2])\|=\max_p \|R_{\ast p}-\Gamma_{\ast p}\|^2\leq k_{\mathrm{max}}.
    $$
    Note that $\|R_{\ast p}-\Gamma_{\ast p}\|^2$ is not random, so the expectation has no effect on this quantity. This yields
    $$v(R-\Gamma)\leq \max\{m\bar{k}/n, k_{\mathrm{max}}\}= m\bar{k}/n,$$ 
    so by the Matrix Bernstein Inequality \cite{tropp2015}, and in light of Equation~\ref{eq:termbound}, we have
    $$\PP(\norm{R-\Gamma}\geq t)\leq(n+m)\,\mathrm{exp}\left(\frac{-t^2/2}{v(R-\Gamma)+\sqrt{k_{\mathrm{max}}}\,t/3}\right).$$
    Then with overwhelming probability as $m\rightarrow\infty$, we get
    $$
    \|R-\Gamma\|=O\left(\sqrt{\frac{m\log(m)\bar{k}}{n}}\right)
    $$
\end{proof}

Proof of Lemma~\ref{lem:eigenH}.
\begin{proof}
The $i,j$ entry of $\mathcal{H}(\EE(RR^\top))$ is given by
\begin{equation*}
    \left[\hol{\EE(RR^\top)}\right]_{ij}=
    \begin{cases}
        0&\text{when }i=j,\\
        \mu_r &\text{when }i\neq j,\text{ but both in class }r,\\
        \sum_{p=1}^m\frac{\tau_{rp}\tau_{sp}}{n_r n_s}&\text{when }i\neq j,\text{ and in different classes }r,s.
    \end{cases}
\end{equation*}
For $r\neq s$, we can see immediately that the $r,s$ block of this matrix is equal to 
$$
\frac{(\mathcal{T}\mathcal{T}^\top)_{r,s}}{n_r n_s} J_{n_r,n_s} = (ZB\Sigma_{\mathcal{T}}BZ^\top)_{r,s}
$$
as claimed. The $r,r$ diagonal block is given by 
$$
\mu_r (J_{n_r}-I_{n_r})= \mu_r (n_r-1)(J_{n_r}/n_r)-\mu_r(I_{n_r}-J_{n_r}/n_r).
$$
The first term equals
$$
\left[\sum_{p=1}^m \tau_{rp}(\tau_{rp}-1)\right] \frac{1}{n_r^2}J_{n_r} = \left[(\mathcal{T}\mathcal{T}^\top)_{rr}- (\mathcal{T}j)_r\right]\frac{1}{n_r^2} J_{n_r}=(ZB\Sigma_{\mathcal{T}}BZ^\top)_{rr},
$$
which completes the proof.
\end{proof}

\begin{lemma}[Spectral norm bound for $\mathcal{H}(RR^\top) -\mathcal{H}(\EE(RR^\top))$]
\label{lem:hollowconcentration}
With overwhelming probability as $m\rightarrow\infty$, $m\geq n$,
$$\|\mathcal{H}(RR^\top) - \mathcal{H}(\EE(RR^\top))\|=O(\sqrt{m\log(m)k_{\mathrm{max}}\bar k}).
$$
\end{lemma}
\begin{proof}
    Letting $R_{\ast p}$ be the $p$-th column of $R$, we may write
    $$\mathcal{H}(RR^\top) - \mathcal{H}(\EE(RR^\top))=\sum_{p=1}^m[\mathcal{H}(R_{\ast p}R_{\ast p}^\top)-\mathcal{H}(\EE(R_{\ast p}R_{\ast p}^\top))],$$
    which is a sum of independent, mean-0 symmetric matrices. Each term is bounded via
    \begin{align*}
    \|\mathcal{H}(R_{\ast p}R_{\ast p}^\top)&-\mathcal{H}(\EE(R_{\ast p}R_{\ast p}^\top))\|^2\\
    &\leq\|R_{\ast p}R_{\ast p}^\top-\EE(R_{\ast p}R_{\ast p}^\top)\|_F^2\\
    &=\sum_{r=1}^d\left[\frac{\tau_{rp}(n_r-\tau_{rp})}{n_r}+\frac{\tau_{rp}(\tau_{rp}-1)(n_r(n_r-1)-\tau_{rp}(\tau_{rp}-1))}{n_r(n_r-1)}\right]+2\sum_{r<s}\frac{\tau_{rp}\tau_{sp}(n_r n_s-\tau_{rp}\tau_{sp})}{n_r n_s}\\
    &\leq \left(\sum_r \tau_{rp}\right)^2\leq k_{\mathrm{max}}^2.
    \end{align*}
    Let $D_p$ be the $n\times n$ diagonal matrix having diagonal entries given by the vector $R_{\ast p}$. Since $R_{ip}\in\{0,1\}$ for all $i,p$, we see that
    \begin{align*}
    [\mathcal{H}(R_{\ast p}R_{\ast p}^\top)^2]_{ij}&=\sum_s (1-\delta_{is})(1-\delta_{sj})R_{ip}R_{sp}R_{jp}=R_{ip}R_{jp}(k_p-R_{ip}-R_{jp})=R_{ip}R_{jp}(k_p-2)\\
    [\mathcal{H}(R_{\ast p}R_{\ast p}^\top)^2]_{ii}&=\sum_s (1-\delta_{is}) R_{ip}R_{sp}=R_{ip}(k_p-R_{ip})=R_{ip}(k_p-1),
    \end{align*}
    where $\delta_{ij}=1$ if $i=j$, and otherwise $\delta_{ij}=0$. So
    $$\mathcal{H}(R_{\ast p}R_{\ast p}^\top)^2=(k_p-2)\mathcal{H}(R_{\ast p}R_{\ast p}^\top)+(k_p-1)D_p=(k_p-2)R_{\ast p}R_{\ast p}^T+D_p.$$
    
    Then we may define the matrix variance statistic as
    \begin{align*}
    v(\mathcal{H}(RR^\top)-\mathcal{H}(\EE(RR^\top)))&=\left\|\sum_{p=1}^m\EE[\mathcal{H}(R_{\ast p}R_{\ast p}^\top-\EE(R_{\ast p}R_{\ast p}^\top))^2]\right\|\\
    &=\left\|\sum_{p=1}^m\left[\EE[\mathcal{H}(R_{\ast p}R_{\ast p}^\top)^2]-\mathcal{H}(\EE(R_{\ast p}R_{\ast p}^\top))^2\right]\right\|\\
    &=\left\|\sum_{p=1}^m\left[(k_p-2)\EE(R_{\ast p}R_{\ast p}^\top)+\EE(D_p)-\mathcal{H}(\EE(R_{\ast p}R_{\ast p}^\top))^2\right]\right\|.
    \end{align*}
    The representation for the matrix inside the spectral norm on the first line makes it clear that this is positive semidefinite. Since $\sum_p \mathcal{H}(\EE(R_{\ast p}R_{\ast p}^\top))^2$ is positive semidefinite, we get
    $$
    v(\mathcal{H}(RR^\top)-\mathcal{H}(\EE(RR^\top)))\leq \left\|\sum_{p=1}^m (k_p-2) \EE[R_{\ast p}R_{\ast p}^\top]+\EE(D_p)\right\|.
    $$
    For this matrix, we have the entrywise bounds for $i\neq j$
    $$
    \sum_{p=1}^m (k_p-2) \EE[R_{i p}R_{jp}]\leq \sum_{p=1}^m \frac{k_p(k_p-2)}{\tilde{c} n},\quad \sum_{p=1}^m (k_p-1)\EE[R_{ip}] \leq \sum_{p=1}^m \frac{k_p(k_p-1)}{\tilde{c}n}. 
    $$ 
    so by \cite{horn2012matrix} Theorem 8.1.18, we obtain
    \begin{align*}
    v(\mathcal{H}(RR^\top)-\mathcal{H}(\EE(RR^\top)))\leq&\left\|\sum_{p=1}^m k_p(k_p/\tilde cn)J_n\right\|\\
    =&\frac{mk_{\mathrm{max}} \bar{k}}{\tilde{c}}.
    \end{align*}
    
    Therefore, according to the Matrix Bernstein Inequality,
    $$
    \PP(\|\mathcal{H}(RR^\top)-\mathcal{H}(\EE(RR^\top))\|\geq t)\leq 2n\,\mathrm{exp}\left(\frac{-t^2/2}{mk_{\mathrm{max}}\bar{k}/\tilde{c}+k_{\mathrm{max}}t/3}\right),
    $$
    so with probability at least $1-m^{-2}$ as $m\rightarrow\infty$, we have
    $$
    \|\mathcal{H}(RR^\top)-\mathcal{H}(\EE(RR^\top))\|\leq
        7\sqrt{m\log(m)k_{\mathrm{max}}\bar{k}/\tilde{c}},
    $$
    since $m\geq n$.
\end{proof}

The proof of Lemma~\ref{lem:eigtrapping} follows immediately by applying Weyl's inequality.

Recall the notation for the following spectral decompositions:
\begin{align*}
  \Gamma&=[U|U_{\perp}][S\oplus 0][V|V_{\perp}]^\top\\
  \hol{RR^\top}&=[\hat U|\hat U_{\perp}][\hat \Lambda\oplus\hat\Lambda_{\perp}][\hat U|\hat U_{\perp}]^\top\\
  \hol{\EE(RR^\top)}&=[\tilde U|\tilde U_{\perp}][\tilde \Lambda\oplus\tilde\Lambda_{\perp}][\tilde U|\tilde U_{\perp}]^\top\\
  \hat U^\top R&=\hat X\hat S\hat V^\top
\end{align*}

\begin{lemma}[Identical actions on $\mathrm{span}(U)$]
\label{lem:identical}
$R$ and $\Gamma$ satisfy the following equations
$$
U^\top(R-\Gamma)=\tilde{U}^\top(R-\Gamma)=0.
$$
\end{lemma}
\begin{proof}
For any class $r$, we have that 
$$
\sum_{i\in C_r}(R_{ip}-\Gamma_{ip})=\tau_{rp}-n_r\frac{\tau_{rp}}{n_r}=0.
$$
Then since the columns of $Z$ are constant within each node class, we must have $Z^\top(R-\Gamma)=0.$ Note that
$$
\Gamma=USV^\top=ZB\mathcal T,
$$
so $\mathrm{span}(U)\subseteq \mathrm{span}(Z)$, and thus $U^\top (R-\Gamma)=0$, too. As stated just after the definition of $\tilde{U}$ in Equation~\ref{eq:tildeu}, $\mathrm{span}(\tilde{U})\subseteq\mathrm{span}(Z)$, so we have $\tilde{U}^\top(R-\Gamma)=0$.
\end{proof}

Let $\tilde U^\top \hat{U}= W_{U_1}\Sigma_ 
 U W_{U_2}^\top$ and $V^\top \hat{V}=W_{V_1}\Sigma_VW_{V_2}^\top$ be their SVD decompositions. Then, let $W_U^*=W_{U_1}W_{U_2}^\top$ and $W_V^*=W_{V_1}W_{V_2}^\top$. Recall the definition of $\Delta$ from Equation~\ref{eq:delta}:
 $$
 \Delta = \min_{r,s} \left|\lambda_r(\sqrt{B}\Sigma_{\mathcal{T}}\sqrt{B})-(-\mu_s)\right|.
 $$
\begin{lemma}
\label{lem:nearorthogonal}
Suppose Assumption~\ref{assump:community} and Equation~\ref{eq:bigdelta} hold. Then with overwhelming probability as $m\rightarrow\infty$,
    \begin{align*}
        \|\tilde{U}^\top \hat{U}-W_U^*\|_F &\leq O\left(\frac{m\log(m)k_{\mathrm{max}}\bar{k}}{\Delta^2}\right),\\
        \|V^\top \hat{V}-W_V^*\|_F&\leq O\left(\frac{m\log(m) \kappa^4 k_{\mathrm{max}}\bar{k}}{\Delta^2}\right),
    \end{align*}
    where $\kappa=\sigma_1(\Gamma)/\sigma_d(\Gamma)$.
\end{lemma}
\begin{proof}
    Let $\theta_1,\ldots, \theta_d$ denote the canonical angles between $\mathrm{span}(\tilde{U})$ and $\mathrm{span}(\hat{U})$. Then by the Davis-Kahan Theorem \cite{yu2015useful} and Lemma~\ref{lem:hollowconcentration}, we see that with overwhelming probability as $m\rightarrow\infty$,
    \begin{align}
    \max_i |\sin(\theta_i)| &=\norm{\sin\Theta(\tilde{U},\hat{U})}\notag\\
    &\leq \frac{2\sqrt{d}\|\hol{RR^\top}-\hol{\EE(RR^\top)}\|}{\Delta}\notag\\
    &\leq \frac{7\sqrt{m\log(m)k_{\mathrm{max}}\bar{k}/\tilde{c}}}{\Delta}.\label{eq:ubound}
    \end{align}
    And (up to a different constant), the same bound holds for $\min_W\|\tilde{U}-\hat{U}W\|,$ where the minimization is over $d\times d$ orthogonal matrices. To obtain the stated bound, we observe that the singular values of $\tilde{U}^\top \hat{U}$ are the cosines of the canonical angles between $\mathrm{span}(\tilde{U})$ and $\mathrm{span}(\hat{U})$. Then 
    \begin{equation}
    \|\tilde{U}^\top \hat{U}-W_U^*\|_F^2 = \sum_{i=1}^d |\cos(\theta_i)-1|^2\leq \sum_{i=1}^d \sin(\theta_i)^4\leq d\max_i |\sin(\theta_i)|^4.
    \label{eq:usquared}
    \end{equation}

    Note that there is an orthogonal matrix $O^*$ such that $\tilde{U}O^*=U$. Since $\hat{V}$ are the right singular vectors of $\hat{U}^\top R$, and $V$ are the right singular vectors of $U^\top\Gamma$, we will apply Davis-Kahan again to compare the right singular vectors. However, since $\hat{U}\approx \tilde{U}W_U^*$, in order to get a close approximation to $U^\top\Gamma$, we will instead use $(\hat{U}(W_U^*)^\top O^*)^\top R$, since this has the same matrix of right singular vectors $\hat{V}$, up to an orthogonal matrix. Then on the same event as above, we may use Lemma~\ref{lem:identical} to argue that there exists an orthogonal matrix $O$ satisfying
    \begin{align*}
        \|\hat{V}O-V\|_F &\leq \frac{2^{3/2}\sqrt{d}(2\|U^\top \Gamma\|+\|(\hat{U}(W_U^*)^\top O^*)^\top R-U^\top \Gamma\|)\|(\hat{U}(W_U^*)^\top O^*)^\top R-U^\top \Gamma\|}{\sigma_d^2(U^\top \Gamma)}\\
        &= \frac{2^{3/2}\sqrt{d}(2\|\Gamma\|+\|((W_U^*)^\top O^*)^\top(\hat{U}-\tilde{U}W_U^*)^\top R\|)\|((W_U^*)^\top O^*)^\top(\hat{U}-\tilde{U}W_U^*)^\top R\|}{\sigma_d^2(\Gamma)}\\
        &=\frac{2^{3/2}\sqrt{d}(2\|\Gamma\|+\|\hat{U}-\tilde{U}W_U^*\|\|R\|)\|\hat{U}-\tilde{U}W_U^*\|\|R\|}{\sigma_d^2(\Gamma)}.
    \end{align*}
    By Lemma~\ref{lem:rgamma}, and under Condition~\ref{eq:bigdelta},
    $$\|R\|\leq \|\Gamma\|+\|R-\Gamma\|\leq 2\|\Gamma\|.$$
    So 
    $$\|\hat{V}O-V\|_F \leq \frac{C\|\Gamma\|^2\|\hat{U}-\tilde{U}W_U^*\|}{\sigma_d^2(\Gamma)},$$ and after applying Equation~\ref{eq:ubound} and the definition of $\kappa$, this becomes
    $$
    \|\hat{V}O-V\|_F \leq C\kappa^2\frac{\sqrt{m\log(m)k_{\mathrm{max}}\bar{k}}}{\Delta}.
    $$
    Since this holds for some orthogonal matrix $O$, it certainly holds for the minimizing orthogonal matrix, which is given by $(W_V^*)^\top$. Arguing as in Equation~\ref{eq:usquared} gives the stated bound.
\end{proof}

\begin{lemma}
\label{lem:intertwining}
Continue to suppose that Assumption~\ref{assump:community} and Equation~\ref{eq:bigdelta} hold, and $m \bar{k}\geq nk_{\mathrm{max}}$. Let $W=(O^*)^\top W_U^* \hat{X}$ be an orthogonal matrix, where each of these factors is defined above. Then
$$
\|SW-W_V^* \hat{S}\|_F \leq O\left(\max\left\{\frac{m^{3/2}\log(m)\kappa^4 (k_{\mathrm{max}} \bar k)^{3/2}}{\sqrt{n}\Delta^2},\frac{m\log(m)\sqrt{k_{\mathrm{max}}}\bar k}{\sqrt{n}\Delta}\right\}\right).
$$
\end{lemma}
\begin{proof}
We have the following matrix expansion:
\begin{align}
    W_V^*\hat S=&(W_V^*-V^\top\hat V)\hat S+V^\top\hat V\hat S\notag\\
    =&(W_V^*-V^\top\hat V)\hat S+V^\top R^\top\hat U\hat X\notag\\
    =&(W_V^*-V^\top\hat V)\hat S+V^\top\Gamma^\top\hat U\hat X+V^\top(R-\Gamma)^\top\hat U\hat X\notag\\
    =&(W_V^*-V^\top\hat V)\hat S+SU^\top \hat U\hat X+V^\top(R-\Gamma)^\top UU^\top\hat U\hat X +V^\top(R-\Gamma)^\top(I-UU^\top)\hat U\hat X\notag\\
    =&(W_V^*-V^\top\hat V)\hat S+S(O^*)^\top W_U^*\hat{X}+S(U^\top \hat U-(O^*)^\top W_U^*)\hat X\notag\\
    &\quad +V^\top(R-\Gamma)^\top UU^\top\hat U\hat X+V^\top(R-\Gamma)^\top(I-UU^\top)\hat U\hat X\notag\\
    =&SW\notag\\
    &+(W_V^*-V^\top\hat V)\hat S\label{eq:intertwining1}\\
    &+S(O^*)^\top(\tilde{U}^\top \hat U- W_U^*)\hat X\label{eq:intertwining2}\\
    &+V^\top(R-\Gamma)^\top UU^\top\hat U\hat X\label{eq:intertwining3}\\
    &+V^\top(R-\Gamma)^\top(I-UU^\top)\hat U\hat X\label{eq:intertwining4}
\end{align}

Since $\|\hat{S}\|\leq \|R\|\leq \|\Gamma\|+\|R-\Gamma\|\leq 2\|\Gamma\|$ by Lemma~\ref{lem:rgamma} and under Condition~\ref{eq:bigdelta}, we can bound this factor by controlling $\|\Gamma\|$. 
$$
\|\Gamma\|=\|U^\top \Gamma\|= \|\sqrt{B}\mathcal{T}\|,
$$
since $U=Z\sqrt{B}O$ for some orthogonal matrix $O$. Since this matrix is $d$-dimensional, there is little loss in bounding this with the Frobenius norm, and we have
\begin{equation}
\|\hat{S}\|^2\leq \|\sqrt{B}\mathcal{T}\|_F^2 = \sum_{p,r} \frac{\tau_{rp}^2}{n_r}\leq \frac{m k_{\mathrm{max}}\bar{k}}{\tilde{c}n}.
\label{eq:sbound}
\end{equation}

Lines~\ref{eq:intertwining1} and \ref{eq:intertwining2} may be bounded with nearly identical arguments, as 
\begin{align*}
    \norm{(W_V^*-V^\top\hat V)\hat S}_F
    &\leq\norm{W_V^*-V^\top\hat V}_F\norm{\hat S}
    \leq O\left(\frac{m\log(m) \kappa^4 k_{\mathrm{max}}\bar{k}}{\Delta^2}\right)\,O\left(\sqrt{\frac{m k_{\mathrm{max}}\bar k}{n}}\right)\\
    &=O\left(\frac{m^{3/2}\log(m)\kappa^4 (k_{\mathrm{max}} \bar k)^{3/2}}{\sqrt{n}\Delta^2}\right).
\end{align*}
Note that the term \ref{eq:intertwining2} does not require the factors of $\kappa$. For the term in Line~\ref{eq:intertwining3}, we apply Lemma~\ref{lem:identical} to see that $U^\top(R-\Gamma)=0$, so this whole term is 0. The last term \ref{eq:intertwining4} is bounded as
\begin{align*}
    \norm{V^\top(R-\Gamma)^\top(I- UU^\top)\hat U\hat X}_F&\leq \norm{R-\Gamma}\norm{(I-UU^\top)\hat U}_F\\
    &\leq O\left(\sqrt{\frac{m\log(m)\bar{k}}{n}}\right)O\left(\frac{\sqrt{m\log(m)k_{\mathrm{max}}\bar{k}}}{\Delta}\right)\\
    &= O\left(\frac{m\log(m)\sqrt{k_{\mathrm{max}}}\bar{k}}{\sqrt{n}\Delta}\right).
\end{align*}
In sum,
$$
\|SW-W_V^* \hat{S}\|_F \leq O\left(\max\left\{\frac{m^{3/2}\log(m)\kappa^4 (k_{\mathrm{max}} \bar k)^{3/2}}{\sqrt{n}\Delta^2},\frac{m\log(m)\sqrt{k_{\mathrm{max}}}\bar k}{\sqrt{n}\Delta}\right\}\right).
$$
\end{proof}

We may also use the expression for $\|\Gamma\|$ in Equation~\ref{eq:sbound} to obtain lower bounds that we will apply later:
\begin{equation}
\label{eq:lowerboundgamma}
\sigma_d^2(\Gamma)\geq \frac{\|\Gamma\|^2}{\kappa^2}\geq  \frac{\|\Gamma\|_F^2}{\kappa^2 d}\geq \frac{m\bar{k}}{\tilde{c} n\kappa^2 d}.
\end{equation}
Since $U^\top R=U^\top \Gamma$, $\sigma_d(R)$ satisfies the same lower bound:
\begin{align*}
\sigma_d^2(R)&=\max_{\substack{Y\in \RR^{n\times d}\\ Y^\top Y=I_d}}\min_{\|v\|=1} v^\top Y^\top RR^\top Yv\\
&\geq \min_{\|v\|=1} v^\top U^\top RR^\top Uv\\
&= \min_{\|v\|=1} v^\top U^\top \Gamma \Gamma^\top Uv\\
&= \sigma_d^2(\Gamma).
\end{align*}
As for $\sigma_d(\hat{U}^\top R),$ we obtain
\begin{align*}
\sigma_d(\hat{U}^\top R)&\geq \sigma_d(\tilde{U}^\top R)-\|(\hat{U}-\tilde{U}W_U^*)^\top R\|\\
&\geq \sigma_d(\Gamma)-\frac{14\sqrt{m\log(m)k_{\mathrm{max}}\bar{k}/\tilde{c}}}{\Delta}\|\Gamma\|\\
&= \sigma_d(\Gamma)\left(1-\frac{14\sqrt{m\log(m)\kappa^2 k_{\mathrm{max}}\bar{k}/\tilde{c}}}{\Delta}\right).
\end{align*}

\begin{lemma}
\label{lem:intertwininginverses}
Suppose the slight strengthening of Equation~\ref{eq:bigdelta}, $\Delta \geq 3b\kappa$. Then
$$
\|S^{-1}W-W_V^*\hat{S}^{-1}\|_F\leq O\left(\frac{\sqrt{m}\log(m)\kappa^5 k_{\mathrm{max}}^{3/2}\sqrt{\bar{k}} d}{\sigma_d(\Gamma)\Delta}\right)
$$
\end{lemma}
\begin{proof}
\begin{align}
S^{-1}W-W_V^*\hat{S}^{-1}&= S^{-2}(SW-S^2W_V^*\hat{S}^{-1})\notag\\
&=S^{-2}(SW-W_V^*\hat{S})+S^{-2}(W_V^*\hat{S}^2-S^2W_V^*)\hat{S}^{-1}.\label{eq:intertwininginverses}
\end{align}

The first term is bounded by 
$$
O\left(\max\left\{\frac{m\log(m)\kappa^5 k_{\mathrm{max}}^{3/2} \bar k\sqrt{d}}{\sigma_d(\Gamma)\Delta^2},\frac{\sqrt{m}\log(m)\kappa \sqrt{k_{\mathrm{max}}\bar k d}}{\sigma_d(\Gamma)\Delta}\right\}\right).
$$

For the second term, we may bound 
$$
\|S^{-2}\|\|\hat{S}^{-1}\|\leq O\left(\frac{n\kappa^2 d}{\sigma_d(\Gamma)m\bar{k}}\right),
$$
in light of Equation~\ref{eq:lowerboundgamma} and the following equations, and using the stronger assumption on $\Delta$ to ensure that $\sigma_d(\hat{U}^\top R)\gtrsim \sigma_d(\Gamma).$ The internal matrix in the second term of Equation~\ref{eq:intertwininginverses} may be expanded as
    \begin{align*}
    W_V^*\hat S^2&=(W_V^*-V^\top\hat V)\hat S^2+V^\top\hat V \hat S^2\\
    &=(W_V^*-V^\top\hat V)\hat S^2+V^\top R^\top \hat{U}\hat{U}^\top R \hat V\\
    &=(W_V^*-V^\top\hat V)\hat S^2+V^\top(R^\top \hat{U}\hat{U}^\top R -\Gamma^\top \Gamma)\hat V+V^\top \Gamma^\top \Gamma\hat V\\
    &=(W_V^*-V^\top\hat V)\hat S^2
    +V^\top(R^\top \hat{U}\hat{U}^\top R -\Gamma^\top \Gamma)\hat V+S^2V^\top\hat V\\
    &=(W_V^*-V^\top\hat V)\hat S^2
    +V^\top(R^\top \hat{U}\hat{U}^\top R -\Gamma^\top \Gamma)\hat V\\
    &\quad +S^2(V^\top\hat V-W_V^*)
    +S^2W_V^*,
    \end{align*}
    so
    $$
    -(S^2W_V^*-W_V^*\hat S^2) =(W_V^*-V^\top\hat V)\hat S^2-S^2(W_V^*-V^\top \hat{V})
    +V^\top(R^\top \hat{U}\hat{U}^\top R -\Gamma^\top \Gamma)\hat V.
    $$
    And thus we have
    \begin{align*}
    \norm{S^2W_V^*-W_V^*\hat S^2}_F \leq&\norm{W_V^*-V^\top\hat V}_F(\|\hat{S}\|^2+\|S\|^2) +\norm{R^\top \hat{U}\hat{U}^\top R -\Gamma^\top \Gamma}\\
    &\leq O\left(\frac{m\log(m)\kappa^4k_{\mathrm{max}}\bar{k}}{\Delta^2}\right)O\left(\frac{m k_{\mathrm{max}}\bar{k}}{n}\right)\\
    &\quad+O\left(\frac{m^{3/2}\sqrt{\log(m)}(k_{\mathrm{max}}\bar{k})^{3/2}}{n\Delta}\right)
    \end{align*}
    So the second term in Equation~\ref{eq:intertwininginverses} is bounded by
    \begin{align*}
    \|S^{-2}(W_V^*\hat{S}^2-S^2 W_V^*)\hat{S}^{-1}\|_F&\leq O\left(\max\left\{\frac{m\log(m)\kappa^6 k_{\mathrm{max}}^2\bar{k} d}{\sigma_d(\Gamma)\Delta^2},\frac{\sqrt{m\log(m)}\kappa^2k_{\mathrm{max}}^{3/2}\sqrt{\bar{k}} d}{\sigma_d(\Gamma)\Delta}\right\}\right)
    \end{align*}
Combining terms and applying the strengthening of Equation~\ref{eq:bigdelta}, we obtain the bound
\begin{align*}
\|S^{-1}W-W_V^*\hat{S}^{-1}\|_F &\leq O\left(\max\left\{\frac{m\log(m)\kappa^6 k_{\mathrm{max}}^2\bar{k} d}{\sigma_d(\Gamma)\Delta^2},\frac{\sqrt{m}\log(m)\kappa^2k_{\mathrm{max}}^{3/2}\sqrt{\bar{k}} d}{\sigma_d(\Gamma)\Delta}\right\}\right)\\
&=O\left(\frac{\sqrt{m}\log(m)\kappa^5 k_{\mathrm{max}}^{3/2}\sqrt{\bar{k}} d}{\sigma_d(\Gamma)\Delta}\right).
\end{align*}
\end{proof}

\begin{theorem}
\label{thm:unscaled2inf}
Suppose Assumption~\ref{assump:community} and the slight strengthening of Equation~\ref{eq:bigdelta}, $\Delta \geq 3b\kappa$, hold. 

Assume 
$$
\rho:= \frac{\sqrt{m}\log(m)\kappa^6 k_{\mathrm{max}}^{3/2}\sqrt{\bar{k}}d}{\Delta}\rightarrow 0.
$$
Then with overwhelming probability as $m\rightarrow\infty$, 
\begin{equation}
\|\hat{V}-VW_V^*\|_{2\rightarrow\infty} \leq O\left(\frac{\rho}{1-\rho}\|V\|_{2\rightarrow\infty}\right) \leq O\left(\frac{\log(m)\kappa^7 k_{\mathrm{max}}^{5/2}d^{3/2}}{\Delta}\right).
\label{eq:unscaledv}
\end{equation}
\end{theorem}
\begin{proof}
Let $W=(O^*)^\top W_U^* \hat{X}$, so 

For $1\leq p\leq m$, 
\begin{align*}
W_V^*(\hat{V}-VW_V^*)^\top e_p&=(W_V^*\hat{V}^\top-V^\top)e_p\\
&=(W_V^*\hat{S}^{-1}\hat{S}\hat{V}^\top-V^\top)e_p\\
&=(S^{-1}W\hat{S}\hat{V}^\top-V^\top)e_p+(W_V^*\hat{S}^{-1}-S^{-1}W)\hat{S}\hat{V}^\top e_p\\
&=S^{-1}((O^*)^\top W_U^* \hat{U}^\top R - U^\top \Gamma)e_p+(W_V^*\hat{S}^{-1}-S^{-1}W)\hat{S}\hat{V}^\top e_p\\
&=S^{-1}(\hat{U}(W_U^*)^\top O^*-U)^\top Re_p +(W_V^*\hat{S}^{-1}-S^{-1}W)\hat{S}\hat{V}^\top e_p.
\end{align*}
The norm of this vector may be bounded by
\begin{align*}
\|(\hat{V}-VW_V^*)^\top e_p\|&\leq \|S^{-1}\|\|\hat{U}-\tilde{U}W_U^*\|\|Re_p\|+\|S^{-1}W-W_V^*\hat{S}^{-1}\|\|\hat{S}\|\|\hat{V}^\top e_p\|\\
&\leq \|S^{-1}\|\|\hat{U}-\tilde{U}W_U^*\|\|\hat{S}\|\|\hat{V}^\top e_p\|+\|S^{-1}W-W_V^*\hat{S}^{-1}\|\|\hat{S}\|\|\hat{V}^\top e_p\|\\
&\lesssim \left(\frac{\sqrt{m\log(m)\kappa^2 k_{\mathrm{max}}^2\bar{k}}}{\Delta}+\frac{\sqrt{m}\log(m)\kappa^6 k_{\mathrm{max}}^{3/2}\sqrt{\bar{k}} d}{\Delta}\right)\|\hat{V}^\top e_p\|\\
&\lesssim \frac{\sqrt{m}\log(m)\kappa^6 k_{\mathrm{max}}^{3/2}\sqrt{\bar{k}}d}{\Delta}(\|V^\top e_p\|+\|(\hat{V}-VW_V^*)^\top e_p\|).
\end{align*}
When the fractional part $\rho\rightarrow 0$, then even with hidden constants, it will eventually be $<1$. In this case,
$$
\|\hat{V}-VW_V^*\|_{2\rightarrow\infty} \leq \frac{\rho}{1-\rho}\|V\|_{2\rightarrow\infty}.
$$

We may control the $2\rightarrow\infty$ norm of $V$ by bounding
$V^\top e_p = S^{-1}U^\top \Gamma e_p = S^{-1}\sqrt{B}\tau_p$ as
$$
\|V^\top e_p\|^2\leq \frac{1}{\sigma_d^2(\Gamma)}\sum_{r=1}^d \frac{\tau_{rp}^2}{n_r}\leq \frac{k_{\mathrm{max}}^2}{\tilde{c}n\sigma_d^2(\Gamma)}.
$$
By Equation~\ref{eq:lowerboundgamma},
$$\sigma_d^2(\Gamma)\geq \frac{m\bar{k}}{\tilde{c} n\kappa^2 d}, $$
so
\begin{equation}
\|V\|_{2\rightarrow\infty} \leq O\left(\frac{\kappa k_{\mathrm{max}}\sqrt{d}}{\sqrt{m\bar{k}}}\right).
\label{eq:v2inf}
\end{equation}
Combining this with the previous bound on $\|\hat{V}-VW_V^*\|_{2\rightarrow\infty}$ yields the stated bound.
\end{proof}

Let us now prove Theorem~\ref{thm:scaledv}.

\begin{proof}

Let $W=(O^*)^\top W_U^* \hat{X}$. 
We may write the matrix in question as
$$
\hat{V}\hat{S}-VSW = (\hat{V}-VW_V^*)\hat{S}+V(W_V^*\hat{S}-SW)
$$
Now we have the bounds for each term given by 
$$\|\hat{V}-VW_V^*\|_{2\rightarrow\infty} \lesssim\frac{\log(m)\kappa^7 k_{\mathrm{max}}^{5/2}d^{3/2}}{\Delta}.$$
$$\|\hat{S}\|\lesssim \sqrt{\frac{m k_{\mathrm{max}} \bar{k}}{n}}$$
$$
\|SW-W_V^* \hat{S}\|_F \leq O\left(\max\left\{\frac{m^{3/2}\log(m)\kappa^4 (k_{\mathrm{max}} \bar k)^{3/2}}{\sqrt{n}\Delta^2},\frac{m\log(m)\sqrt{k_{\mathrm{max}}}\bar k}{\sqrt{n}\Delta}\right\}\right).
$$
$$
\|V\|_{2\rightarrow\infty} \leq O\left(\frac{\kappa k_{\mathrm{max}}\sqrt{d}}{\sqrt{m\bar{k}}}\right)
$$

We get
$$
\|\hat{V}\hat{S}-VSW\|_{2\rightarrow\infty} \lesssim \frac{\sqrt{m}\log(m)\kappa^7 k_{\mathrm{max}}^3\sqrt{\bar{k}} d^{3/2}}{\sqrt{n}\Delta}.
$$
So for $\sqrt{n}\|\hat{V}\hat{S}-VSW\|_{2\rightarrow\infty}\rightarrow 0$, we need
$$\Delta^2 \gg m\log^2(m)\kappa^{14}k_{\mathrm{max}}^6 \bar{k} d^3.$$
This condition is sufficient for Theorem~\ref{thm:unscaled2inf} to hold.

To compare $S V_p$ to $S V_{p'}$, observe that $SV^\top(e_p-e_{p'})=U^\top \Gamma (e_p-e_{p'})= W \sqrt{B}\mathcal{T}(e_p-e_{p'}),$ since $U=Z\sqrt{B} W$ for some orthogonal matrix $W$, and $\sqrt{B}Z^\top Z\sqrt{B}=I.$ Then whenever $\tau_p\neq \tau_{p'}$, we have:
$$
\|SV_p - SV_{p'}\|=\|\sqrt{B}\mathcal{T}(e_p-e_{p'})\|^2=\sum_r \frac{(\tau_{rp}-\tau_{rp'})^2}{n_r}\geq \frac{1}{\tilde{c}n}.
$$
Comparing this to the size of the noise, measured in $2\rightarrow\infty$ norm, we have that with overwhelming probability as $m\rightarrow\infty$, 
$$
\sqrt{n}\|\hat{V}\hat{S}-VSW\|_{2\rightarrow\infty} \rightarrow 0.
$$

Consider applying hierarchical clustering to the rows of $X\in \RR^{m\times d}$ using complete linkage, and let $x_p\in\RR^d$ denote the $p$th row of $X$. The initial clustering corresponds to all singleton clusters, $\{1\}, \{2\}, \ldots, \{m\}$. After several steps, we have a collection of clusters $C_1^{(i)},\ldots, C_{n_i}^{(i)}\subseteq [m]$. To obtain the new collection of clusters $C_1^{(i+1)},\ldots, C_{n_{i+1}}^{(i+1)},$ find the pair of clusters $C_k^{(i)}, C_{\ell}^{(i)}$ minimizing the quantity 
$$ 
d(C_k^{(i)},C_{\ell}^{(i)})=\max\{\|x_p-x_s\|: p\in C_k^{(i)}, s\in C_{\ell}^{(i)}\}
$$
and merge them to get $C_1^{(i+1)}$. The remaining clusters are equal to the clusters $C_j^{(i)}$, listed in any order, so $n_{i+1}=n_i - 1.$ When $m$ is large enough, with high probability, we have that $\sqrt{n}\|\hat{V}\hat{S}-VSW\|_{2\rightarrow\infty} \leq 1/(10\tilde{c}n)$. When $\tau_p=\tau_s$, $e_p^\top VSW= e_s^\top VSW$, so 
$$
\|(e_p^\top-e_s^\top)\hat{V}\hat{S}\|\leq \|e_p^\top(\hat{V}\hat{S}-VSW)\|+\|e_s^\top(\hat{V}\hat{S}-VSW)\|\leq 2\|\hat{V}\hat{S}-VSW\|_{2\rightarrow\infty}\leq 1/(5\tilde{c}n).
$$
When $\tau_p\neq \tau_s$, $\|(e_p-e_s)^\top VSW\|\geq 1/(\tilde{c}n),$ so
\begin{align*}
\|(e_p-e_s)^\top \hat{V}\hat{S}\|&\geq \|(e_p-e_s)^\top VSW\|-\|(e_p-e_s)^\top(\hat{V}\hat{S}-VSW)\|\\
&\geq 1/(\tilde{c}n)-2\|\hat{V}\hat{S}-VSW\|_{2\rightarrow\infty}\\
&\geq 4/(5\tilde{c}n).
\end{align*}
Then at any step in the hierarchical clustering algorithm, if there remains two clusters $C_k^{(i)}, C_{\ell}^{(i)}$, such that all of the corresponding type vectors are equal, 
$$
d(C_k^{(i)}, C_{\ell}^{(i)}) \leq 1/(\tilde{c}n).
$$
Meanwhile, any pair of clusters $C_k^{(i)}, C_{\ell}^{(i)}$ such that their union contains $p,s$ with $\tau_p\neq \tau_s$ must satisfy 
$$
d(C_k^{(i)},C_{\ell}^{(i)})\geq 4/(5\tilde{c}n).
$$
This proves that the initial sequence of merge steps will always unite clusters of interactions that all share a common type vector. Indeed, once a first merge step $i+1$ unites two clusters of interactions which do not all share a common type vector, this implies that $C_1^{(i)},\ldots, C_{n_i}^{(i)}$ cluster the type vectors perfectly.
\end{proof}

\section{Additional Simulation Results}
\label{sec:additional}

All of the simulation settings are identical to Section~\ref{sec:simulation}, but we include additional figures showing the growth or decrease of intermediate quantities bounded in Appendix~\ref{sec:proofs}.

\begin{figure}[ht]
\begin{tabular}{cc}
  \includegraphics[width=0.5\textwidth]{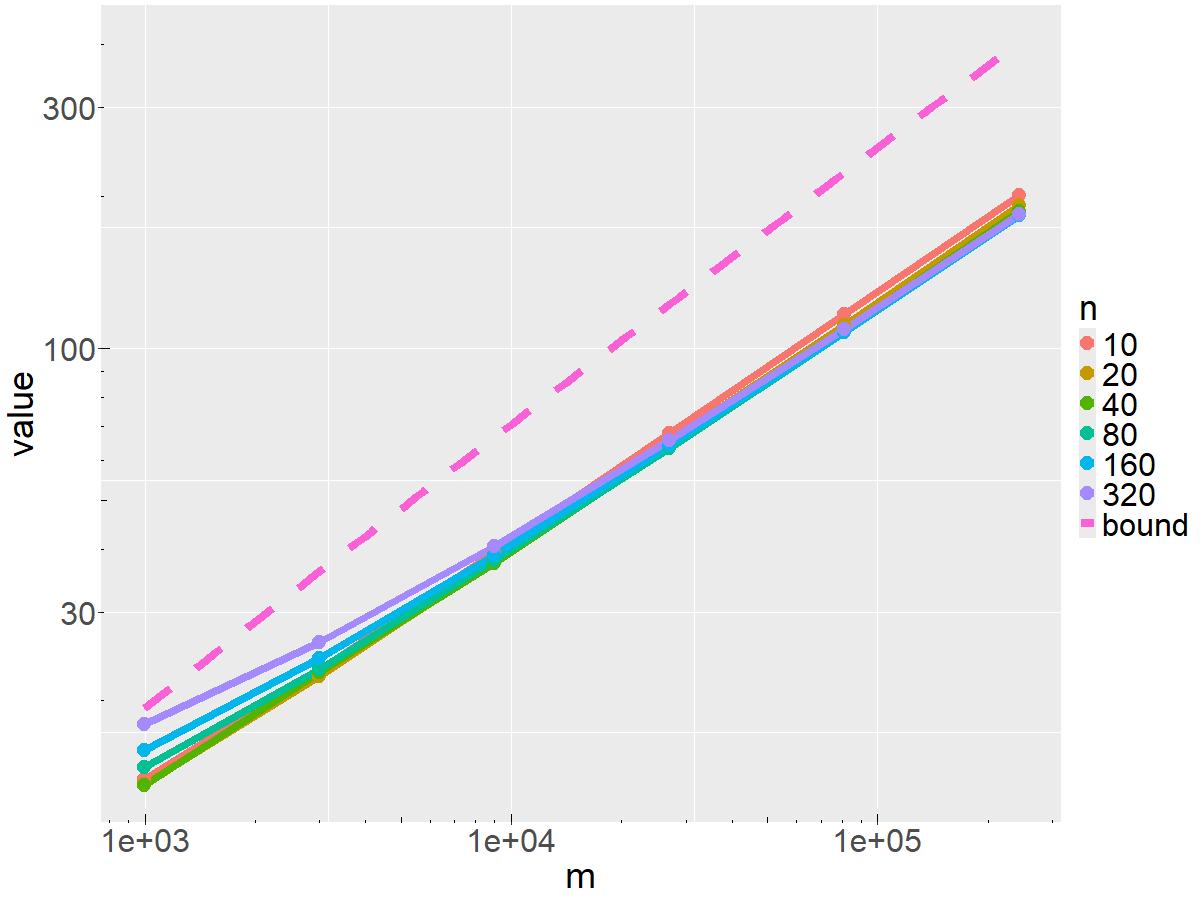} &   \includegraphics[width=0.5\textwidth]{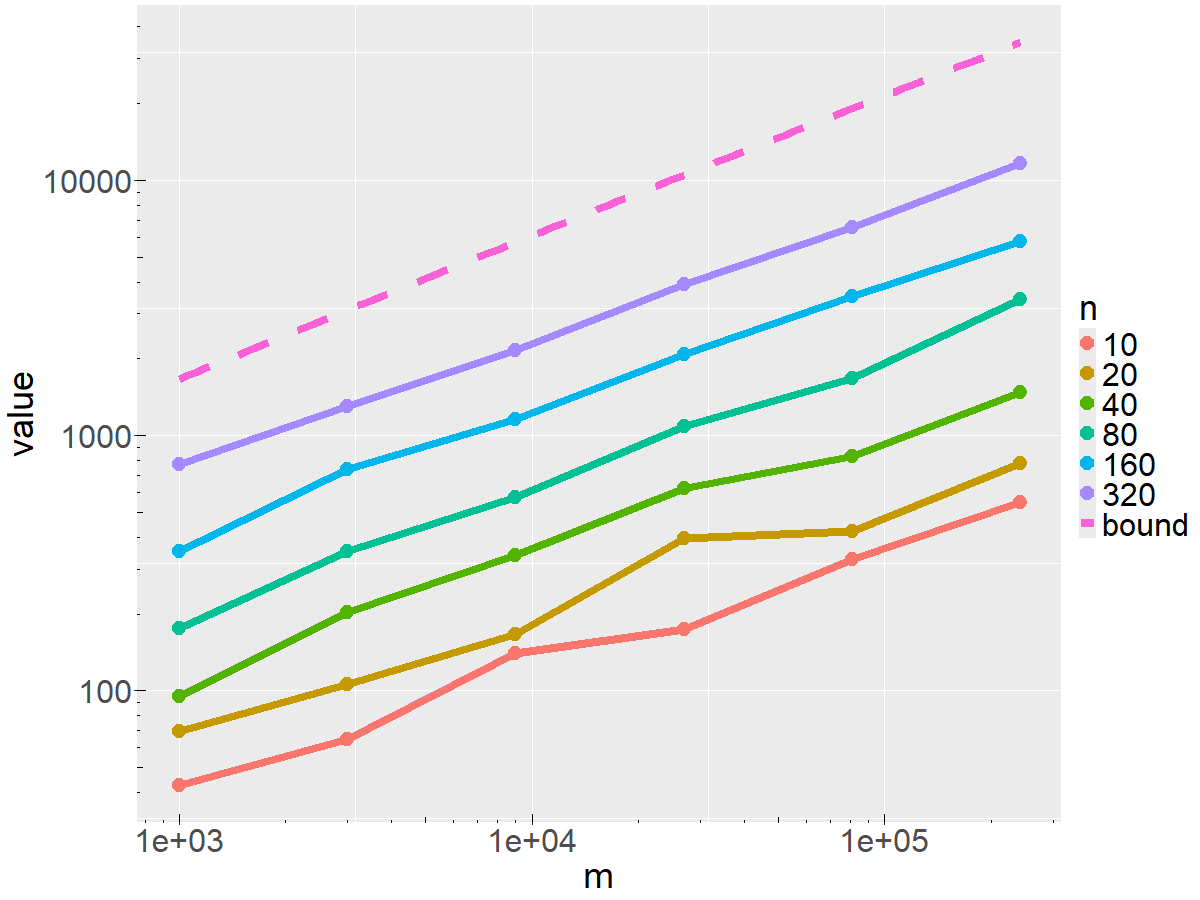} \\
(a)$\norm{R-\Gamma}$ & (b) $\norm{\hol{RR^\top}-\hol{\EE(RR^\top)}}$ \\
 \includegraphics[width=0.5\textwidth]{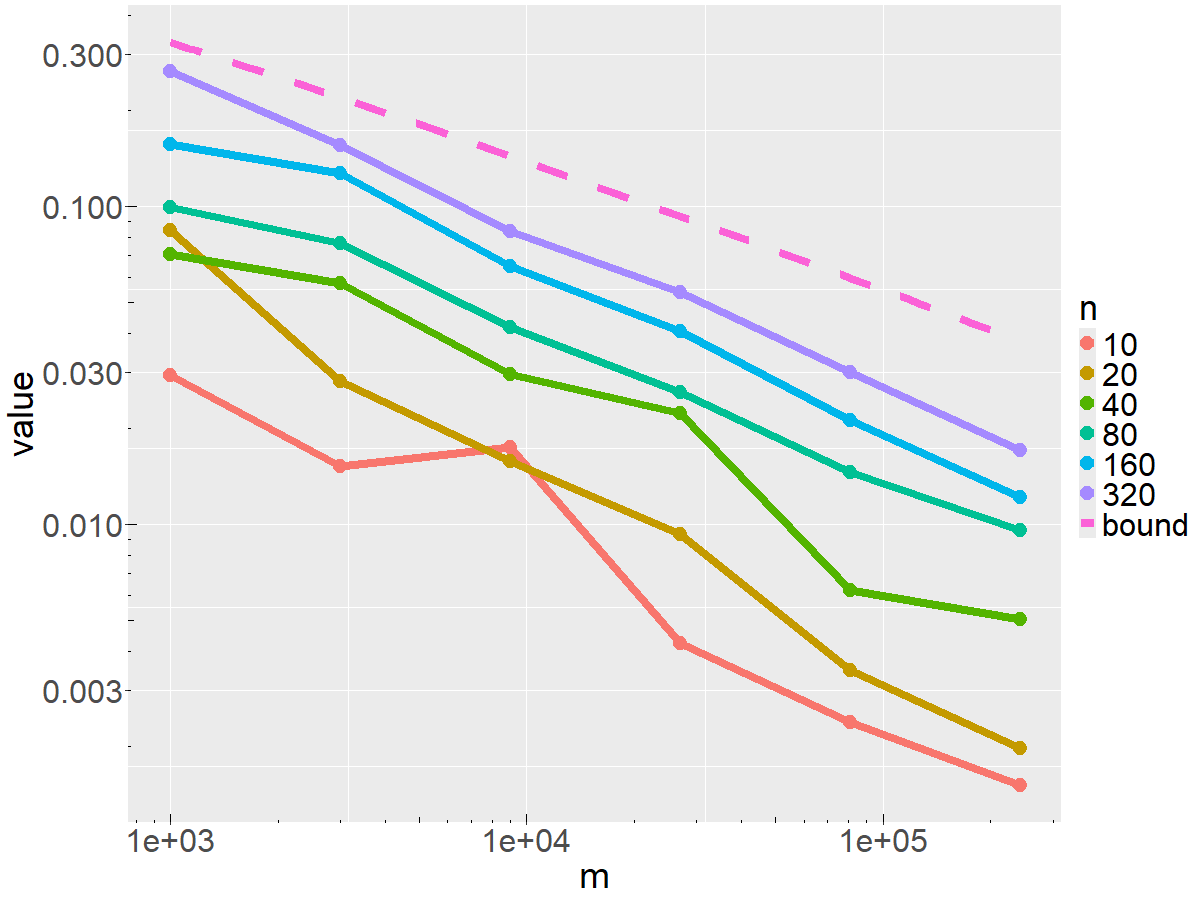} &   \includegraphics[width=0.5\textwidth]{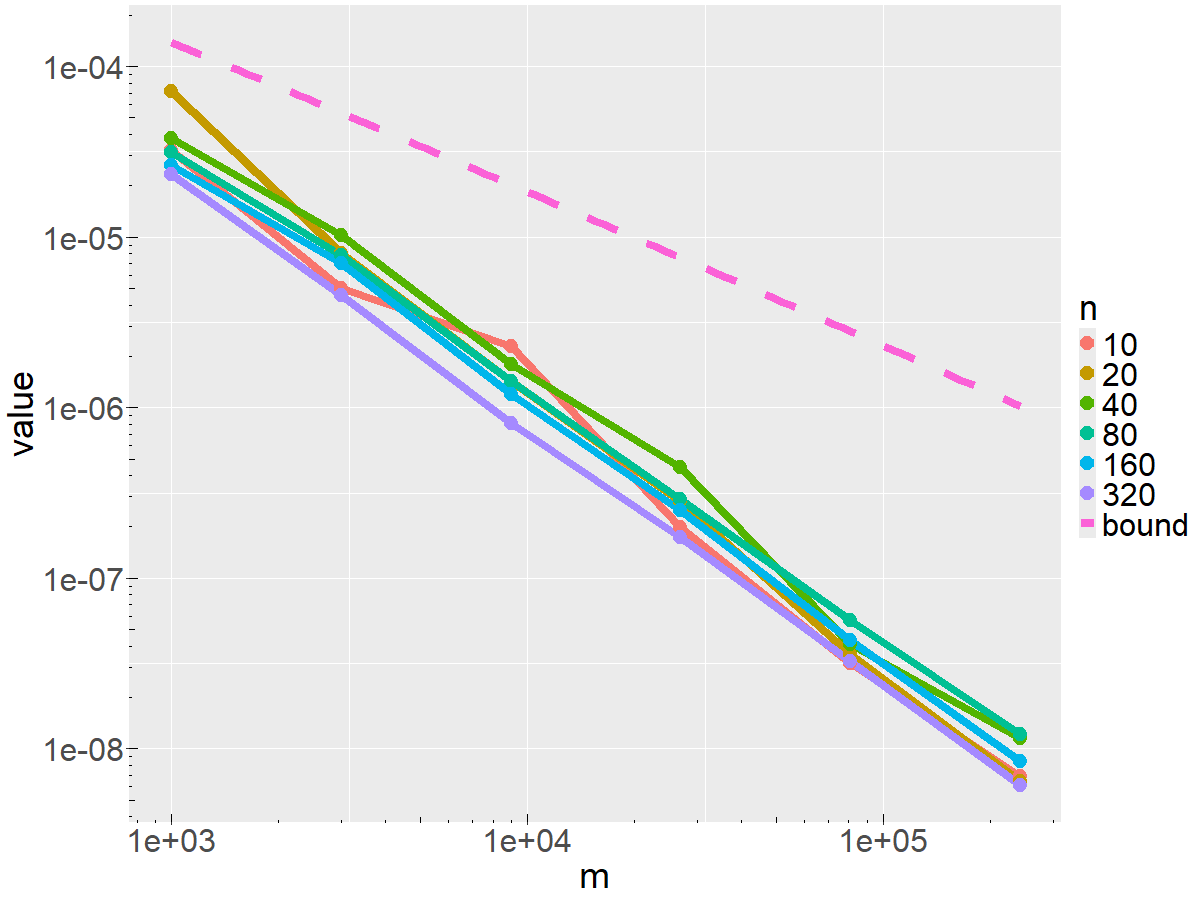} \\
(c)$\norm{SW-W_V^*\hat S}_F$ & (d)  $\norm{S^{-1}W-W_V^*\hat S^{-1}}_F$ \\
\includegraphics[width=0.5\textwidth]{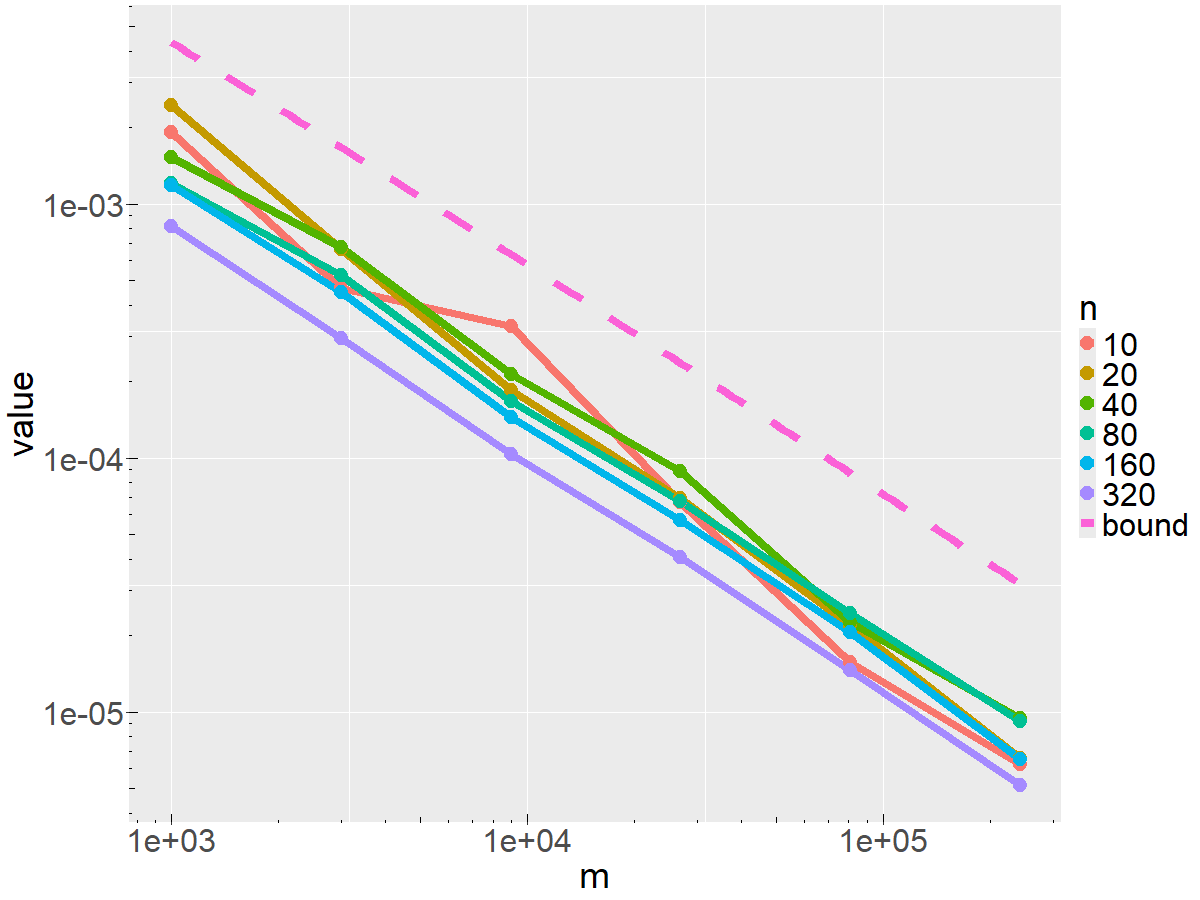} &   \includegraphics[width=0.5\textwidth]{images/sim_022825_VhatShat_VSW_2Inf_log.png} \\
(e)$\norm{\hat V-VW_V^*}_{2\rightarrow\infty}$ & (f)  $\norm{\hat V\hat S-VSW}_{2\rightarrow\infty}$ \\
\end{tabular}
\caption{\label{fig:IH-sim-growing-k} Simulation results under different $m,n$ with $k_{\mathrm{max}}=n/2$.}
\end{figure}

\begin{figure}[ht]
\begin{tabular}{cc}
  \includegraphics[width=0.5\textwidth]{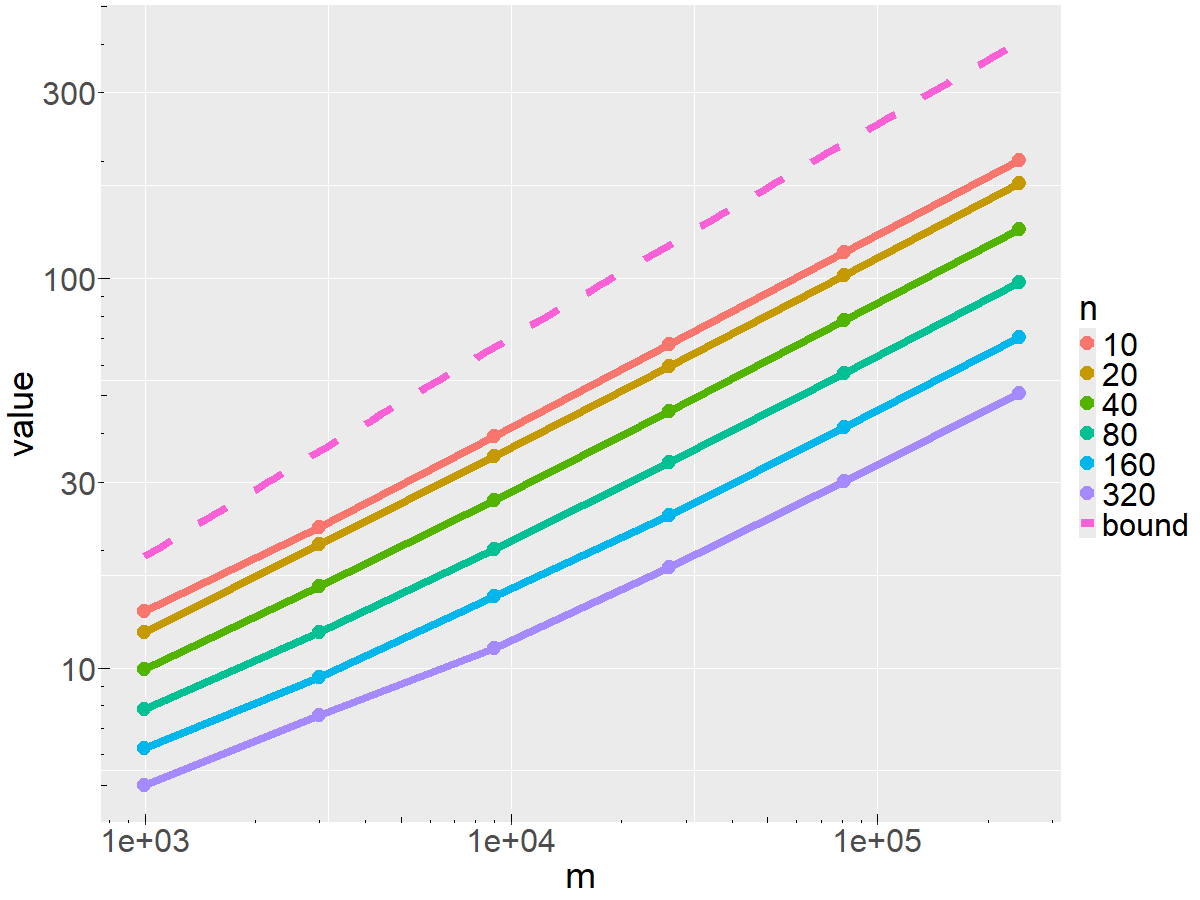} &   \includegraphics[width=0.5\textwidth]{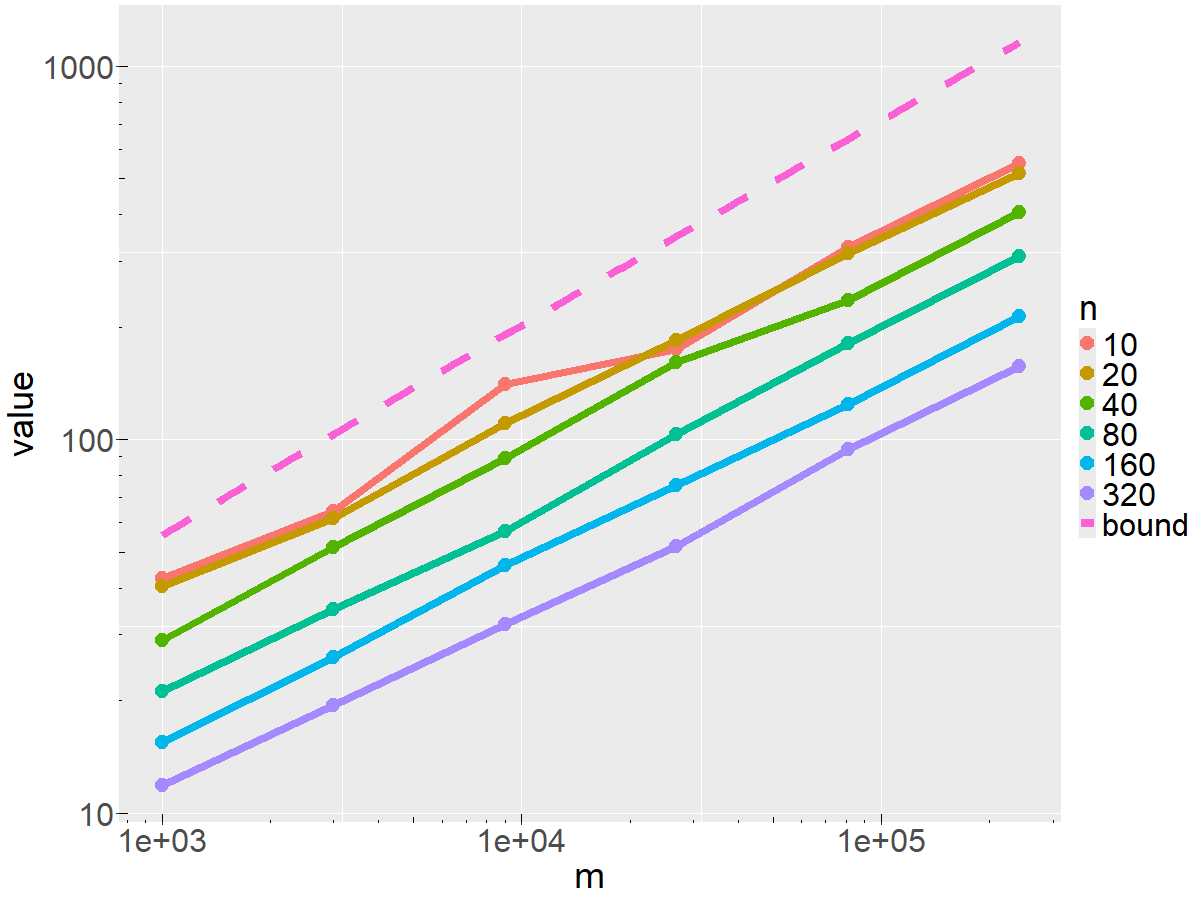} \\
(a)$\norm{R-\Gamma}$ & (b) $\norm{\hol{RR^\top}-\hol{\EE(RR^\top)}}$ \\
 \includegraphics[width=0.5\textwidth]{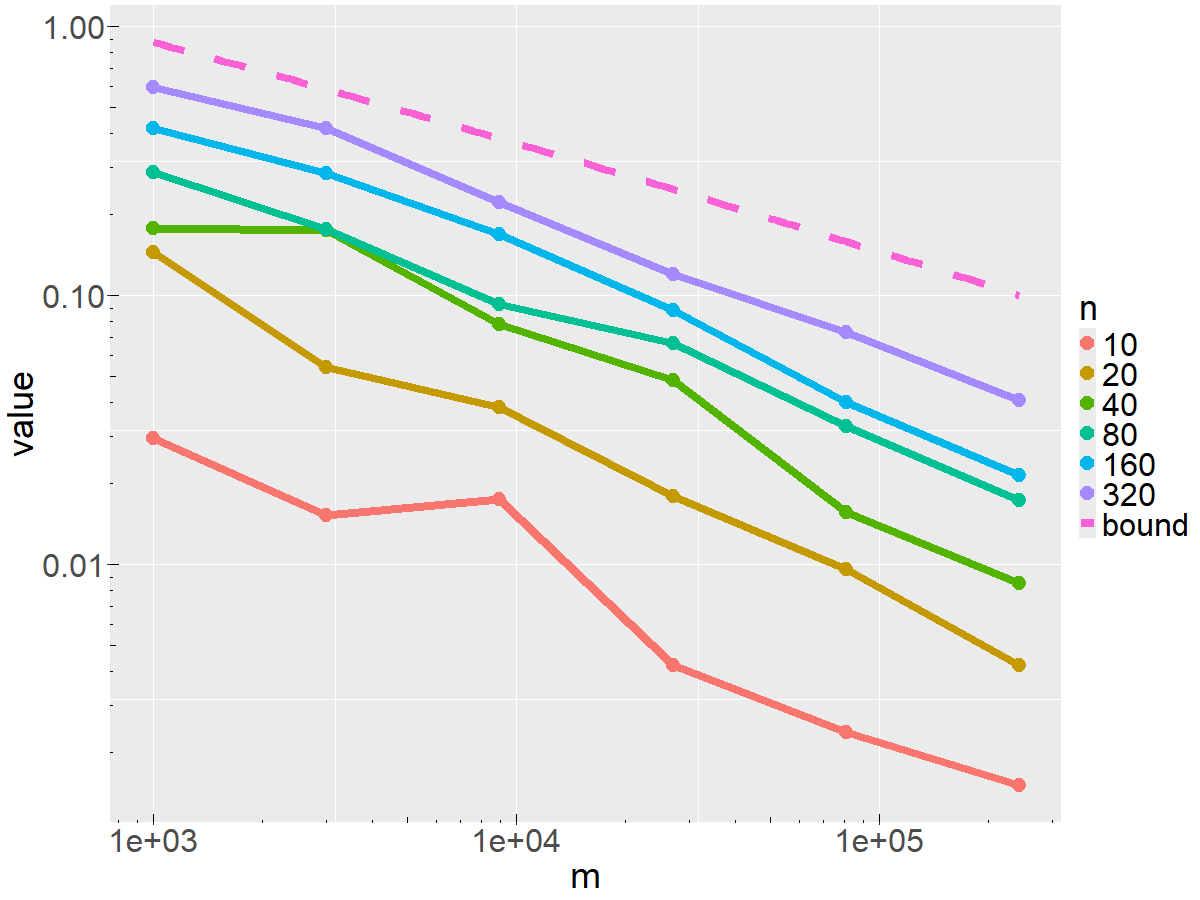} &   \includegraphics[width=0.5\textwidth]{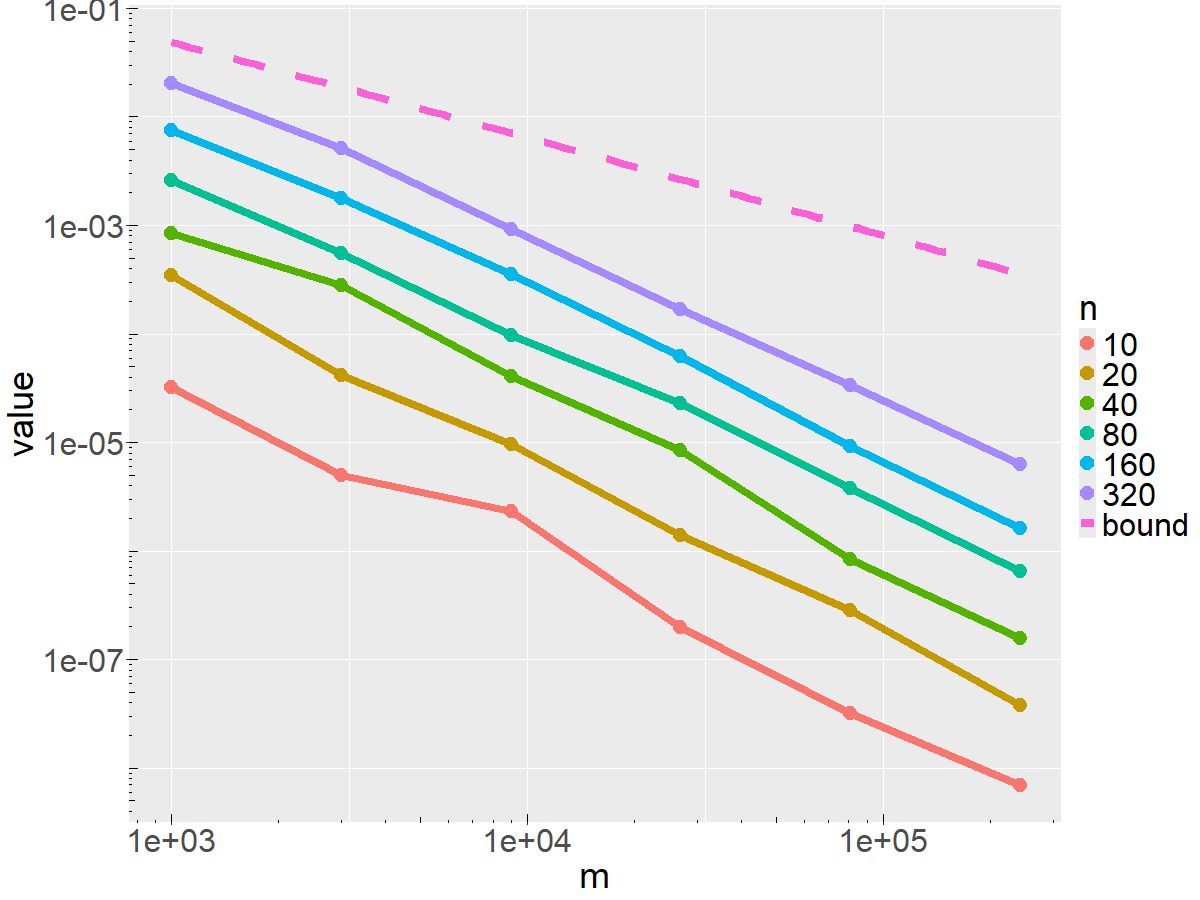} \\
(c)$\norm{SW-W_V^*\hat S}_F$ & (d)  $\norm{S^{-1}W-W_V^*\hat S^{-1}}_F$ \\
\includegraphics[width=0.5\textwidth]{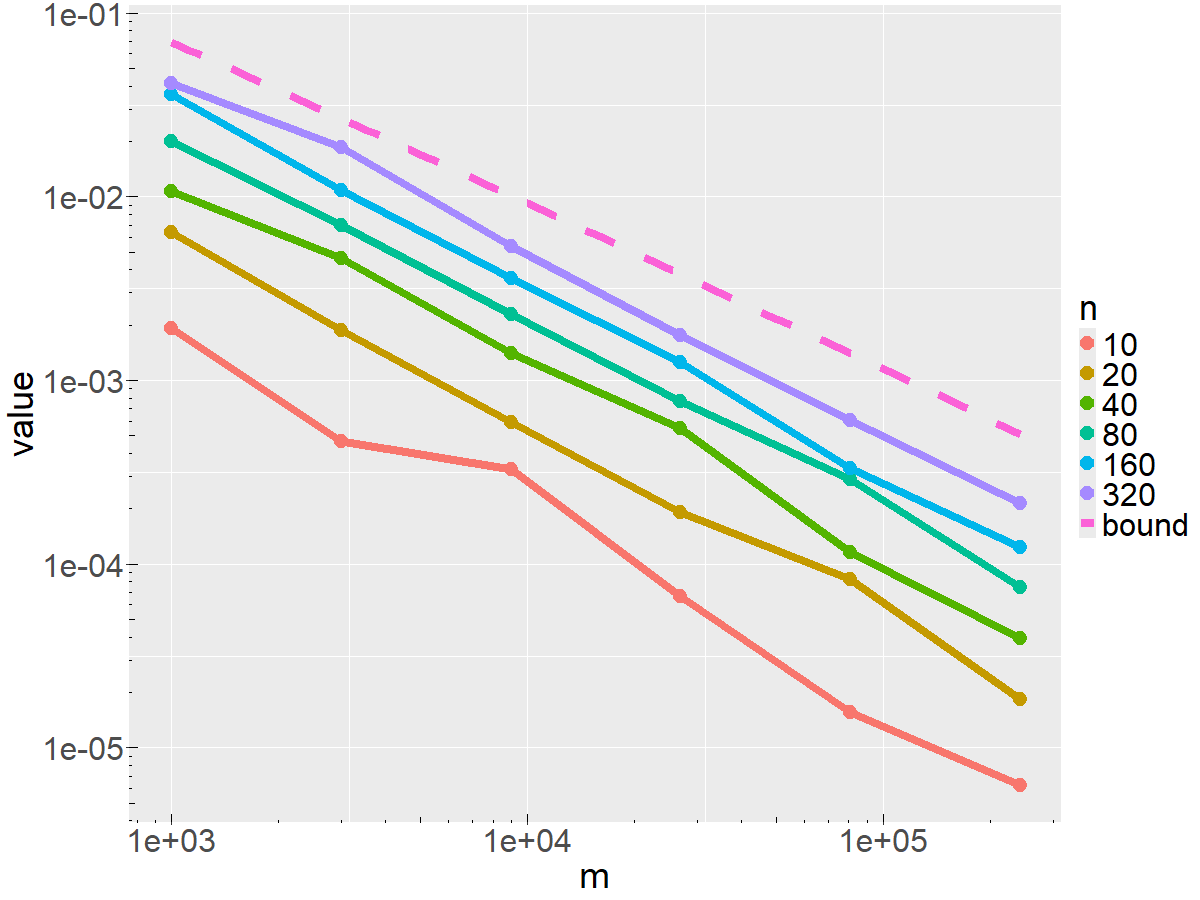} &   \includegraphics[width=0.5\textwidth]{images/sim_022825_fixed_k_VhatShat_VSW_2Inf_log.png} \\
(e)$\norm{\hat V-VW_V^*}_{2\rightarrow\infty}$ & (f)  $\norm{\hat V\hat S-VSW}_{2\rightarrow\infty}$ \\
\end{tabular}
\caption{\label{fig:IH-sim-fixed-k} Simulation results under different $m,n$ with $k_{\mathrm{max}}=5$.}
\end{figure}

\end{document}